\documentclass[reqno, 11pt]{amsart}

\usepackage{amsfonts,amssymb,amsbsy,amsmath,amsthm,enumerate}
\usepackage{txfonts}
\usepackage{eucal}
\usepackage{dsfont}
\usepackage{mathrsfs}

\usepackage{bbm}

\usepackage[%condensed,
math]{iwona}
\usepackage{calligra}
\usepackage{frcursive}

\usepackage[T1]{fontenc}

\usepackage{stackrel}

\usepackage[small,nohug,
]{diagrams}
\diagramstyle[labelstyle=\scriptstyle]
\newarrow{Corresponds}<--->

\usepackage[dvips]{graphicx}
\usepackage{color}

\newtheorem{theorem}{Theorem}[section]
\newtheorem{proposition}[theorem]{Proposition}
\newtheorem{lemma}[theorem]{Lemma}
\newtheorem{corollary}[theorem]{Corollary}
\newtheorem{assumption}[theorem]{Assumption}
\newtheorem{definition}[theorem]{Definition}
\newtheorem{remark}[theorem]{Remark}

\numberwithin{equation}{section}

\numberwithin{figure}{section}
\numberwithin{table}{section}

\newcommand\beq{\begin{equation}}
\newcommand{\bea}{\begin{eqnarray}}
\newcommand{\eea}{\end{eqnarray}}
\newcommand{\beas}{\begin{eqnarray*}}
\newcommand{\eeas}{\end{eqnarray*}}
\newcommand{\beql}{\begin{equation} \label}
\newcommand{\eeq}{\end{equation}}

\newcommand{\R}{\mathbb R}
\newcommand{\N}{\mathbb N}
\newcommand{\C}{\mathbb C}                           %NUMBER SETS

\newcommand{\Z}{\mathbb Z}

\newcommand{\s}[1]{\CMcal{#1}}
\newcommand{\f}[1]{\mathcal{#1}}                  %FONTS
\newcommand{\bb}[1]{\mathscr{#1}}
\newcommand{\rr}[1]{\mathfrak{#1}}
\newcommand{\n}[1]{\mathds {#1}}
\newcommand{\expo}[1]{{\rm e}^{#1}}                 %SPECIAL FUNCTION
\newcommand{\dd}{\,{\rm d}}
\newcommand{\ii}{\,{\rm i}\,}
\newcommand{\ncint}{\mathrel{{\ooalign{$\int$\cr\kern+.07em\raise.15ex\hbox{$\pmb{\scriptstyle-}$}\cr}}}}           \newcommand{\ncpartial}{\mathrel{{\ooalign{$\partial$\cr\kern+.29em\raise.79ex\hbox{$\pmb{\scriptstyle-}$}\cr}}}}

\newcommand{\hooklongrightarrow}{\lhook\joinrel\longrightarrow}

\DeclareMathOperator{\Tr}{Tr}

\newcommand{\virg}[1]{\lq\lq#1\rq\rq}                %TIPOGRAPHIC
\newcommand{\ie}{{\sl i.\,e.\,}}
\newcommand{\eg}{{\sl e.\,g.\,}}
\newcommand{\cf}{{\sl cf.\,}}
\newcommand{\etc}{{\sl etc.\,}}

\begin{document}
%------------------------------------------------------------------------------------------------------------------------%
%                                                                        heading
%------------------------------------------------------------------------------------------------------------------------%

\title[Differential geometric invariants for time-reversal symmetric Bloch-bundles]{
Differential geometric invariants for time-reversal symmetric Bloch-bundles II:\\ The low dimensional \virg{Quaternionic} case}

%-----%

\author[G. De~Nittis]{Giuseppe De Nittis}
\address[De~Nittis]{Facultad de Matem\'aticas \& Instituto de F\'{\i}sica,
Pontificia Universidad Cat\'olica de Chile,
Santiago, Chile}
\email{gidenittis@mat.uc.cl}

\author[K. Gomi]{Kiyonori Gomi}
\address[Gomi]{Department of Mathematical Sciences, Shinshu University,  Nagano, Japan}
\email{kgomi@math.shinshu-u.ac.jp}



%-----%

\begin{abstract}
\vspace{-4mm}
This paper is devoted to the construction of differential geometric invariants for the  classification of 
\virg{Quaternionic} vector bundles.
Provided that the base space is a smooth manifold of dimension two or three endowed with
an involution that leaves fixed only a finite number of points, it is possible to prove that the \emph{Wess-Zumino term} and the \emph{Chern-Simons invariant}
yield  topological quantities 
able to distinguish between inequivalent realization of \virg{Quaternionic} structures.
\end{abstract}

%-----%

\maketitle

\vspace{-5mm}
\tableofcontents

%--------------------%
%--------------------%
%--------------------%
\section{Introduction}\label{sect:intro}
The present paper continues the study of the classification of \emph{\virg{Quaternionic} vector bundles}  initiated in \cite{denittis-gomi-14-gen,denittis-gomi-16,denittis-gomi-18}. 
The main novelty with respect to the previous papers consists of the use of differential geometric invariants to classify inequivalent isomorphism classes of \virg{Quaternionic} structures.
 In this sense, as expressed by the title, this paper represents a continuation of \cite{denittis-gomi-15} where  differential geometric techniques have been used to classify 
\virg{Real} vector bundles. 

\medskip

At a topological level, \virg{Quaternionic} vector bundles, or \emph{Q-bundles} for short, are complex vector bundles defined over spaces with {involution} and endowed with a further structure at the level of the total space. An {involution} $\tau$ on a topological space $X$ is a homeomorphism of period 2, \ie  $\tau^2={\rm Id}_X$. The pair 
 $(X,\tau)$ will be called an  {involutive space}. The \emph{fixed point} set of the {involutive space}  $(X,\tau)$
 is  by definition
 $$
X^{\tau}\,:=\, \{x\in X\ |\ \tau(x)=x\}\,.
$$ 
A Q-bundle over $(X,\tau)$ is a pair $(\bb{E},\Theta)$ where $\bb{E}\to X$ denotes the underlying complex vector bundle and $\Theta:\bb{E}\to\bb{E}$ is an \emph{anti}-linear map which covers the action of $\tau$ on the base space and such that $\Theta^2$ acts fiberwise as the multiplication by $-1$. A more precise description  is given in Definition \ref{defi:Q_VB}.
Q-bundles have been introduced for the first time by  J.~L. Dupont in \cite{dupont-69} (under the name of symplectic vector bundle).
They form a {category} of topological objects which is significantly different from the category of complex vector bundles. For this reason the problem of the classification of  Q-bundles over a given involutive space requires the use of tools which are  structurally different  from those usually used in the classification of complex vector bundles. The aim of the present work is to define some differential geometric  invariants   able to distinguish the elements of ${\rm Vec}^m_Q(X,\tau)$ where the latter symbol denotes the set of isomorphism classes of rank $m$ Q-bundles over $(X,\tau)$.

\medskip

The interest for the classification of Q-bundles has increased in the last years because of the connection with the  study of  \emph{topological insulators}. Although this
 work does not focus on the theory of topological insulators (the interested reader is referred to the recent reviews  \cite{hasan-kane-10,ando-fu-15}),  it is worth mentioning that the first example of topological effects in condensed matter related to a \virg{Quaternionic} structure dates back to the seminal works
 by 
L. Fu,  C. L. Kane and E. J. Mele  \cite{kane-mele-05,fu-kane-mele-95}.
The existence of distinguished topological phases for the so-called Kane-Mele model is the result of the simultaneous presence of two symmetries. The first symmetry is given by the invariance of the system under spatial translations. This fact allows the  use of the \emph{Bloch-Floquet theory} \cite{kuchment-93} for the analysis of the spectral properties of the system. As a result,  a well-established procedure provides the construction of a vector bundle, usually known as \emph{Bloch-bundle}, from each gapped energy band of the system.
Even though the details of the construction of the Bloch-bundle will be omitted in this work (the interested reader is referred to \cite{panati-07} or to \cite[Section 2]{denittis-gomi-14}  and references therein) it is important to remark that the Bloch-bundle is a complex vector bundle over the torus $\n{T}^d\simeq\R^d/(2\pi\Z)^d$ as a base space. The integer $d$ represents the dimensionality of the system and the physically relevant dimensions are $d=2,3$. The second crucial ingredient for the topology of the Kane-Mele model is the fermionic (or odd) \emph{time-reversal symmetry} (TRS). In terms of the Bloch-bundle the TRS translates into the involution $\tau_{TR}:\n{T}^d\to\n{T}^d$ of the base space given by $\tau_{TR}(k_1,\ldots,k_d):=(-k_1,\ldots,-k_d)$
and into an anti-linear map $\Theta$ of the total space such that  $\Theta^2=-1$ fiberwise. 
Therefore, one concludes that the different topological phases of the Kane-Mele model are labeled  by the inequivalent realization of Q-bundles over the torus $\n{T}^d$ with involution $\tau_{TR}$, namely by the distinct elements of  ${\rm Vec}^m_Q(\n{T}^d,\tau_{TR})$.

\medskip

The classification of the topological phases of the Kane-Mele model given in \cite{kane-mele-05,fu-kane-mele-95} is summarized below:
\begin{equation}\label{eq:intro1}
{\rm Vec}^m_Q(\n{T}^d,\tau_{TR})\;=\;
\left\{
\begin{aligned}
&\Z_2&\qquad&\text{if}\;\; d=2\\
&\Z_2\;\oplus\;(\Z_2)^3&\qquad&\text{if}\;\; d=3\\
\end{aligned}
\right.
\end{equation}
where $\Z_2:=\{\pm 1\}$ is the the cyclic group of order 2 presented in the multiplicative notation. The topological classification \eqref{eq:intro1}
has been rigorously derived with the use of different techniques in various papers (see \eg \cite{porta-graf-13,denittis-gomi-14-gen,fiorenza-monaco-panati-14}) and generalized to any (low-dimensional) involutive space $(X,\tau)$   in \cite{dos-santos-lima-filho-04,lawson-lima-filho-michelsohn-05} and  in \cite{denittis-gomi-16,denittis-gomi-18}, independently. However, the topological classification based on the construction of homotopy invariants (such as characteristic classes) has the disadvantage of being difficult to compute. 
For this reason one is naturally induced to look for different types of invariants.

\medskip

A special role in the classification of complex vector bundles is played by the \emph{Chern classes}. The latter, in view of the  Chern-Weil homomorphism, can be represented via differential forms and integrated over suitable cocycles. The resulting \emph{Chern numbers} are enough to provide a complete classification of the complex vector bundles in several situations of interest in condensed matter. This is, for instance, the case of the \emph{Quantum Hall Effect} and the related \emph{TKNN numbers} \cite{thouless-kohmoto-nightingale-nijs-82}. Using this observation as the Ariadne's thread one expects to find differential and integral invariants able to classify Q-bundles at least under some reasonable hypothesis. Indeed, \virg{gauge-theoretic invariants} have already been used to reproduce the classification \eqref{eq:intro1}. 
The first pioneering works in this sense are \cite{fu-kane-06,qi-hughes-zhang-08,essin-moore-vanderbilt-09,wang-qi-zhang-10}
where the  Chern-Simons field
theory has been used to relate the  topological phases of the 
Kane-Mele model in 2+1 and 3+1 space-time dimensions with integral quantities like the (time reversal) \emph{polarization}.
Afterwards, these results have been revisited and  put in a more solid mathematical background in various works like \cite{freed-moore-13,gawedzki-15,gawedzki-17,carpentier-delplace-fruchart-gawedzki-15,carpentier-delplace-fruchart-gawedzki-tauber-15,kaufmann-li-wehefritz-16,monaco-tauber-17},  just to mention some of them. If one ignores the differences due to the use of
distinct  mathematical techniques, it is possible to recognize a common message from all  the papers listed above: The topological phases of the two-dimensional Kane-Mele  model are governed by \emph{Wess-Zumino term} \cite{freed-95,gawedzki-99,gawedzki-17} while in the three-dimensional case the relevant object is the \emph{Chern-Simons invariant} \cite{freed-95,gawedzki-17,hu-01}.
The present work is inspired  by the latter considerations and is aimed to provide a general and rigorous description  of the relation between the  Wess-Zumino term, or the Chern-Simons invariant, and the topological
classification of Q-bundles. The main achievements will be presented below. 

\medskip

The two-dimensional case will be described first. In this case the relevant family of base spaces of interest is restricted  by the following:

\begin{definition}[Oriented two-dimensional FKMM-manifold]
\label{def:good_manif_d=2}
An \emph{oriented two-dimensional FKMM-manifold} is a pair $(\Sigma,\tau)$ which meets the following conditions: 
\begin{itemize}
\item[(a')] $\Sigma$ is an oriented two-dimensional compact  Hausdorff manifold  without boundary;
\vspace{1mm}
\item[(b')] The involution $\tau$ preserves the manifold structure and the orientation;
\vspace{1mm} 
\item[(c')] The fixed point set $\Sigma^\tau\neq\emptyset$ consists 
 of a finite collection of points.
\end{itemize}
\end{definition}%

\medskip

\noindent
An example of oriented two-dimensional FKMM-manifold is provided by the torus $\n{T}^2$ with the involution $\tau_{\rm TR}.$
The set of oriented two-dimensional FKMM-manifolds forms a sub-class of the \emph{FKMM-spaces} as defined in Definition \ref{def_FKMM-space}
below
. Q-bundles over these spaces are completely classified by a characteristic class called \emph{FKMM-invariant} (\cf Theorem \ref{theo:biject_K_FKMM}).

\medskip

The crucial result for the classification of 
Q-bundles over two-dimensional FKMM-manifolds is expressed by the following chain of isomorphisms
\begin{equation}\label{eq:intro_02}
{\rm Vec}^{2m}_{Q}(\Sigma, \tau)\;\stackrel{\imath_1}{\simeq}\;[\Sigma,\n{SU}(2)]_{\Z_2}\;/\;[\Sigma,\n{U}(1)]_{\Z_2}\;\stackrel{\imath_2}{\simeq}\;\Z_2\;.
\end{equation}
The first isomorphism $\imath_1$ is essentially proved in Theorem \ref{theo:main-SU2_charat} for $m=1$ and justified in Remark \ref{rk:higer-rk} for  every $m\in\N$.
Elements of  $[\Sigma,\n{SU}(2)]_{\Z_2}$ are
$\Z_2$-homotopy equivalent \footnote{
Let $(X_1,\tau_1)$ and $(X_2,\tau_2)$ be two involutive spaces. A map $f:X_1\to X_2$ is \emph{equivariant} if $f\circ \tau_1=\tau_2\circ f$. Two equivariant maps $f$ and $f'$
are \emph{$\Z_2$-homotopy equivalent} 
 if there exists an equivariant map $\widehat{f} :X_1\times[0,1]\to X_2$ such that $\widehat{f}|_{X \times \{ 0 \}} = f$ and $\widehat{f}|_{X \times \{ 1 \}} = f'$. The involution on  $X_1 \times [0, 1]$ is fixed by $(x, t) \mapsto (\tau_1(x), t)$. This notion provides an equivalence relation
 on the set of equivariant maps, and the set of equivalence classes is denoted with $[X_1, X_2]_{\Z_2}$.}
maps 
 $\xi:\Sigma\to\n{SU}(2)$ constrained by the 
 \emph{equivariance} condition
$\xi(\tau(x))=\xi(x)^{-1}$ for all $x\in \Sigma$. The set 
$[\Sigma,\n{U}(1)]_{\Z_2}$ consists of classes of $\Z_2$-homotopy equivalent  maps 
$\phi:X\to\n{U}(1)$ such that $\phi(\tau(x))=\overline{\phi(x)}=\phi(x)^{-1}$. 
The action of $[\Sigma,\n{U}(1)]_{\Z_2}$ over $[\Sigma,\n{SU}(2)]_{\Z_2}$ is specified in the statement of Theorem \ref{theo:main-SU2_charat}. The second isomorphism $\imath_2$ is described in Section \ref{subsec:altern_FKMM-2D} and is given by the  composition of two identification: The first 
$$
 [\Sigma,\n{SU}(2)]_{\Z_2}\;/\;[\Sigma,\n{U}(1)]_{\Z_2}\;\stackrel{\Phi_\kappa}{\simeq}\;{\rm Map}(\Sigma^\tau,\{\pm1\})\;/\;[\Sigma,\n{U}(1)]_{\Z_2}\;,
$$
proved in Proposition \ref{prop:classification}, shows that the \virg{new} description of Q-bundles in terms of maps $\xi:\Sigma\to \n{SU}(2)$ agrees with the  \virg{old} description in terms of the FKMM-invariant
given in Proposition \ref{prop:fkmm-inv_fkmm-space}. The second identification 
$$
{\rm Map}(\Sigma^\tau,\{\pm1\})\;/\;[\Sigma,\n{U}(1)]_{\Z_2}\;\stackrel{\Pi}{\simeq}\;\Z_2
$$
is described in Theorem \ref{prop:Fu-Kane-Mele_formula1} and is induced by the  \emph{product sign} map (also known as \emph{Fu-Kane-Mele index}).

\medskip

The isomorphism $\imath_1$ in \eqref{eq:intro_02} expresses the fact that an element of ${\rm Vec}^{2m}_{Q}(\Sigma, \tau)$ can be completely identified with an equivariant map $\xi:\Sigma\to\n{SU}(2)$
that,  in many situations, can be built explicitly (\cf Remark \ref{rk:map-frame}). Therefore, the relevant question is whether there is a way to access directly the isomorphism $\imath_2$ from the knowledge of the classifying map $\xi$ without passing through the FKMM-invariant and the {product sign} map. The answer is positive.  First of all it is important to point out that, without loss of generality, the map $\xi$ can be chosen smooth. This allows to define the \emph{Wess-Zumino term}
\begin{equation}\label{eq:itro_inv1}
\f{WZ}_{\Sigma}(\xi)\; :=\;  
-\frac{1}{24\pi^2} \int_{X_\Sigma} {\rm Tr}\left(\widetilde{\xi}^{-1}\cdot\dd\widetilde{\xi}\right)^3\;\qquad\text{\rm mod.}\ \ \Z\;
\end{equation}
where $X_\Sigma$ is any compact three-dimensional oriented  manifold whose boundary coincides with $\Sigma$  and $\widetilde{\xi} : X_\Sigma  \to \n{SU}(2)$ is any extension of $\xi$ (see Definition \ref{dfn:WZ_without_boundary} for more details). The first main result of this paper  is:
\begin{theorem}
\label{teo:main1}
Let $(\Sigma,\tau)$ be an oriented two-dimensional FKMM-manifold in the sense of Definition \ref{def:good_manif_d=2}. Let $(\bb{E},\Theta)$ be a Q-bundle of rank $2m$ over $(\Sigma,\tau)$ and  $\xi\in{\rm Map}(\Sigma,\n{SU}(2))_{\Z_2}$  any map which represents 
$(\bb{E},\Theta)$ in the sense of the isomorphism $\imath_1$ in \eqref{eq:intro_02}.
Then the map
$$
{\rm Vec}^{2m}_{Q}(\Sigma, \tau)\;\ni\;[(\bb{E},\Theta)]\;\longrightarrow\;\expo{\ii\; 2\pi \f{WZ}_{\Sigma}(\xi)}\;\in\;\Z_2
$$
provides a realization of the isomorphism 
${\rm Vec}^{2m}_{Q}(\Sigma, \tau)\simeq\Z_2$ in \eqref{eq:intro_02}.
\end{theorem}

\medskip

\noindent
The proof of Theorem \ref{teo:main1} is postponed to Section \ref{sec:Wess_Zumono_term_class_D2}.   Theorem \ref{teo:main1} clearly applies to the classification of Q-bundles over the involutive torus $(\n{T}^2,\tau_{TR})$ reproducing, in this way, results already existing in the literature. In this regard, let us  give a special mention to the result \cite[Theorem, eq. (2.9)]{gawedzki-17}, previously announced in \cite[II.25, p. 19]{gawedzki-15}. The latter is in agreement with  Theorem \ref{teo:main1} above in view of the equality $\expo{\ii 2\pi \f{WZ}_{\Sigma}(w)} = \expo{\ii 2\pi \f{WZ}_{\Sigma}(\xi)}$ (justified by the Polyakov-Wiegmann formula, \cf Lemma \ref{lem:Polyakov_Wiegmann}) where  the map $w$ employed in \cite{gawedzki-15} is related to the map $\xi$ in  Theorem \ref{teo:main1} by the relation $w = \xi Q$, with $Q$ the constant matrix in \eqref{eq:Q-mat}. However, it is worth pointing out that the validity of  Theorem \ref{teo:main1} goes far beyond the standard case $(\n{T}^2,\tau_{TR})$.
For instance, Theorem \ref{teo:main1} extends the classification of Q-bundles over
 Riemann surfaces of genus $g$ endowed with an orientation-preserving involution with a finite set of fixed points  \cite[Appendix A]{denittis-gomi-14-gen} and this  application seems to be new in the literature.

\medskip

In order to describe the three-dimensional case it is worth mentioning that any Q-bundle $(\bb{E},\Theta)$ over the involutive space $(X,\tau)$ can be equivalently described by  a \emph{principal Q-bundle} $(\bb{P},\hat{\Theta})$ over the same base space (see Section \ref {sect:Q_princ_bund&FKMM}) and that for {principal} Q-bundles there exists a notion of \emph{equivariant Q-connection} (see Section \ref {sect:Q-&connect}). Given a Q-connection 
$\omega\in \Omega^1_{Q}(\bb{P},\rr{u}(2m))$ one can define the associated 
\emph{Chern-Simons 3-form}  
$$
 \f{CS}(\omega)\;:=\;\frac{1}{8\pi^2}\;{\rm Tr}\left(\omega\wedge\dd\omega\;+\;\frac{2}{3}\omega\wedge\omega\wedge\omega\right)
$$
and the \emph{intrinsic Chern-Simons invariant}
\begin{equation}\label{eq:itro_inv2}
\rr{cs}(\bb{P},\hat{\Theta})\;:=\;\int_Xs^*\f{CS}(\omega)\;\qquad\text{\rm mod.}\ \ \Z\;
\end{equation}
according to Definition \ref{def:CS_inv} and Definition \ref{def:intrinsC-S_inv}. Remarkably,
the quantity in the right-hand side of \eqref{eq:itro_inv2} is independent of the choice of the invariant connection $\omega$ or of the global section $s:X\to\bb{P}$ and so defines an invariant for the underlying principal Q-bundle $(\bb{P},\hat{\Theta})$ or equivalently for the associated Q-bundle $(\bb{E},{\Theta})$.

\medskip

It is also necessary to recall that when $(X,\tau)$ is a  three-dimensional FKMM-manifold in the sense of  Definition \ref{def_FKMM-space}
then Proposition \eqref{prop:fkmm-inv_fkmm-space} applies and the following isomorphism holds true:
$$
\begin{aligned}
{\rm Vec}^{2m}_{Q}(X, \tau)&\;\stackrel{\kappa}{\simeq}\; {\rm Map}(X^\tau,\{\pm1\})\;/\;[X,\n{U}(1)]_{\Z_2}& \quad&
\end{aligned}\qquad \forall\ m\in\N\;.
$$
 In the formula above ${\rm Map}(X^\tau,\{\pm1\})\simeq{\Z_2}^{|X^\tau|}$ denotes the set of maps from $X^\tau$ to $\{\pm1\}$ (recall  that $X^\tau$ is a set of finitely many points). 
The group action of $[X,\n{U}(1)]_{\Z_2}$ on ${\rm Map}(X^\tau,\{\pm1\})$ is given by multiplication
and restriction. The map $\kappa$ which implements the isomorphism is  the {FKMM-invariant} (see Section \ref{subsec:fkmm-invariant} and references therein). Given a Q-bundle 
$(\bb{E},\Theta)$ over $(X,\tau)$, its 
FKMM-invariant $\kappa(\bb{E},\Theta)$ can be represented in terms of a map $\phi\in {\rm Map}(X^\tau,\{\pm1\})$ and one can use 
the
product sign to define the so-called \emph{strong} {Fu-Kane-Mele index}
$$
\kappa_{\rm s}(\bb{E},\Theta)\;:=\;\Pi[\phi]\;=\;\prod_{j=1}^{|X^\tau|}\phi(x_j)\;\in\;\Z_2.
$$
It turns out that the definition above is well-posed in the sense that $\kappa_{\rm s}(\bb{E},\Theta)$ only depends on the equivalence class of $\phi$ in  ${\rm Map}(X^\tau,\{\pm1\})\;/\;[X,\n{U}(1)]_{\Z_2}$ and so it defines a topological invariant for $(\bb{E},\Theta)$. This fact is a consequence of the second main result of this paper:

\begin{theorem}
\label{prop:Fu-Kane-Mele=chern_d=3}
Let $(X,\tau)$ be a three-dimensional FKMM-manifold in the sense of  Definition \ref{def_FKMM-space} such that $X^\tau\neq\emptyset$.
Assume in addition that:
\begin{itemize}
\item[(e)] $X$  is oriented and $\tau$ \emph{reverses} the orientation.
\end{itemize}
Let  $(\bb{E},\Theta)$ be a Q-bundle over $(X,\tau)$ with FKMM-invariant $\kappa(\bb{E},\Theta)\in {\rm Map}(X^\tau,\{\pm 1\})/[X,\n{U}(1)]_{\Z_2}$ according to Proposition \ref{prop:fkmm-inv_fkmm-space}. For a given representative  $\phi\in {\rm Map}(X^\tau,\{\pm1\})$ of $\kappa(\bb{E},\Theta)$ let $\Pi[\phi]$ be the associated
{product sign} map. Then, independently of the choice of  $\phi$, it holds true that
\begin{equation}\label{eq:topo_CS_d=3_equality}
\expo{\ii\; 2\pi\rr{cs}(\bb{P},\hat{\Theta})}\;=\;\Pi[\phi]
\end{equation}
where $(\bb{P}, \hat{\Theta})$ is the principal  Q-bundle associated to $(\bb{E},\Theta)$ and $\rr{cs}(\bb{P},\hat{\Theta})$
is the intrinsic Chern-Simons invariant defined in Definition \ref{def:intrinsC-S_inv}.
\end{theorem}

\medskip

\noindent
The proof of Theorem \ref{prop:Fu-Kane-Mele=chern_d=3} is postponed to Section \ref{sect:class_d=3}. Along with Corollary 
\ref{cor:fin_paper}, it expresses the fact that the {strong} {index} 
\begin{equation}\label{eq:topo_CS_d=3_equality_BIS}
\kappa_{\rm s}(\bb{E},\Theta)\;=\;\expo{\ii\; 2\pi\rr{cs}(\bb{P},\hat{\Theta})}\;.
\end{equation}
 is a topological invariant which allows, at least  partially, to classify  Q-bundles. 
In the case of the involutive torus $(\n{T}^3,\tau_{TR})$ described by \eqref{eq:intro1} the invariant $\kappa_{\rm s}(\bb{E},\Theta)$ takes values in the first  (strong) summand of  $\Z_2\oplus(\Z_2)^3$.

\medskip

Theorem \ref{teo:main1} and Theorem \ref{prop:Fu-Kane-Mele=chern_d=3} show that the differential geometric {gauge invariants}
\eqref{eq:itro_inv1} and \eqref{eq:itro_inv2}
can be used as  tools for the classification of Q-bundles in dimension $d=2$ and $d=3$, provided that the base space meets some restrictive conditions.
The results contained in Theorem \ref{teo:main1} and Theorem \ref{prop:Fu-Kane-Mele=chern_d=3} are valid for base spaces much more general than the involutive tori $(\n{T}^d,\tau_{\rm TR})$ usually considered in the 
 literature. However these results   are still
 not completely satisfactory 
 in view of the   restrictions on the nature of the base space that is necessary to assume. There are two questions which are still open and that it would be interesting to answer: \emph{Is it possible to extend Theorem \ref{teo:main1} and Theorem \ref{prop:Fu-Kane-Mele=chern_d=3} to involutive base spaces $(X,\tau)$ such that $X^\tau$ is a sub-manifold of dimension bigger than zero?} 
In the case of Theorem \ref{teo:main1}, \emph{is it  possible to construct the classifying map $\xi$ directly from the projection which represent the Q-bundle in K-theory without relying on the use of a predetermined global frame?}

%-----------------%
\medskip
\noindent
{\bf Structure of the paper.}
The paper is organized as follows: {\bf Section \ref{sect:Q-VB_topological}} contains general facts about the topological classification of Q-bundles. The subsections from {\bf $\S$\ref{sect:basic_def}} to {\bf  $\S$\ref{subsec:fu-kane-mele-invariant}} contain review material while
the last two subsections {\bf  $\S$\ref{subsec:altern_present}} and {\bf  $\S$\ref{subsec:altern_FKMM-2D}} contain a new topological classification for Q-bundles over 
oriented two-dimensional FKMM-manifolds.
{\bf Section \ref{sect:geom_class}} is devoted to the differential geometric aspects of the paper. Subsections {\bf  $\S$\ref{sect:Q_princ_bund&FKMM}} and {\bf  $\S$\ref{sect:Q-&connect}
} contain review material. Subsections 
{\bf  $\S$\ref{sect:chern-simons}},
{\bf  $\S$\ref{sec:Wess_Zumono_term}}  and 
{\bf  $\S$\ref{sec:Wess_Zumono_term_no-bound}}
are focussed on the description of the gauge invariants \eqref{eq:itro_inv1} and \eqref{eq:itro_inv2}. Finally, Subsections 
{\bf  $\S$\ref{sec:Wess_Zumono_term_class_D2}} and 
{\bf  $\S$\ref{sect:class_d=3}} contain the proofs of Theorem \ref{teo:main1} and Theorem \ref{prop:Fu-Kane-Mele=chern_d=3}, respectively.

%-----------------%
\medskip
\noindent
{\bf Acknowledgements.} 
GD's research is supported
 by
the  grant \emph{Iniciaci\'{o}n en Investigaci\'{o}n 2015} - $\text{N}^{\text{o}}$ 11150143 funded  by FONDECYT.	 KG's research is supported by 
the JSPS KAKENHI Grant Number 15K04871.
GD wants to thank the the Erwin Schr\"{o}dinger International Institute for Mathematics and Physics (ESI) of Vienna where the results described in this paper were presented for the first time during the  thematic programme
 \emph{\virg{Topological phases of quantum matter}} held in 2014.
%KG is indebted to 
\medskip
%-----------------%

%--------------------%
%--------------------%
%--------------------%
\section{\virg{Quaternionic} vector bundles from a topological perspective}\label{sect:Q-VB_topological}

In this section base spaces will be considered only from a topological point of view.
Henceforth, we will assume that: 
\begin{assumption}[$\Z_2$-CW-complex]\label{ass:top}
$X$ is a 
topological space which admits the structure of a $\Z_2$-CW-complex. The \emph{dimension} $d$ of $X$ is, by definition, the maximal dimension of its cells 
and $X$ is called \emph{low-dimensional} if $0\leqslant d\leqslant 3$. 
\end{assumption}
\noindent
For the sake of completeness, let us recall that an involutive space $(X,\tau)$ has the structure of a $\Z_2$-CW-complex if it admits a skeleton decomposition given by gluing cells of different dimensions which carry a $\Z_2$-action. 
For a precise definition of the notion of 
$\Z_2$-CW-complex the reader can refer to  \cite[Section 4.5]{denittis-gomi-14} or \cite{matumoto-71,allday-puppe-93}. Assumption  \eqref{ass:top} allows the space $X$ to be made by several disconnected component. However, in the case of multiple components, we will tacitly assume that vector bundles built over $X$ possess fibers of constant rank on the whole base space. Let us recall that a space with a CW-complex structure is automatically Hausdorff and paracompact and it is compact exactly when it is made by a finite number of cells \cite{hatcher-02}. Almost all the examples considered in this paper will concern with spaces with a finite CW-complex structure.

%------------%
\subsection{Basic facts about \virg{Quaternionic} vector bundles}
\label{sect:basic_def}
In this section we recall some basic facts about the topological category of \virg{Quaternionic} vector bundles. Furthermore,
the necessary notation for the description of the various results will be fixed. We refer to \cite{dupont-69,denittis-gomi-14-gen,denittis-gomi-16,denittis-gomi-18} for a more systematic presentation of the subject. 

\medskip

\begin{definition}[{\virg{Quaternionic}} vector bundles]\label{defi:Q_VB}
A {\virg{Quaternionic}} vector bundle, or Q-bundle, over  $(X,\tau)$
is a complex vector bundle $\pi:\bb{E}\to X$ endowed with a (topological) homeomorphism $\Theta:\bb{E}\to \bb{E}$
 such that:
\begin{itemize}

\item[$(Q_1)$] The projection $\pi$ is \emph{equivariant} in the sense that $\pi\circ \Theta=\tau\circ \pi$;
\vspace{1mm}
\item[$(Q_2)$] $\Theta$ is \emph{anti-linear} on each fiber, \ie $\Theta(\lambda p)=\overline{\lambda}\ \Theta(p)$ for all $\lambda\in\C$ and $p\in\bb{E}$ where $\overline{\lambda}$ is the complex conjugate of $\lambda$;
\vspace{1mm}
\item[$(Q_3)$] $\Theta^2$ acts fiberwise as the multiplication by $-1$, namely $\Theta^2|_{\bb{E}_x}=-\n{1}_{\bb{E}_x}$.
\end{itemize}
\end{definition}

\noindent
Let us recall that it
is always possible to endow
$\bb{E}$ with
an essentially unique \emph{equivariant} Hermitian metric $\rr{m}$ with respect to which $\Theta$ is an \emph{anti-unitary} map between conjugate fibers
\cite[Proposition 2.5]{denittis-gomi-14-gen}. In this case equivariant means that 
$$
\rr{m}\big(\Theta(p_1),\Theta(p_2)\big)\,=\,\rr{m}\big(p_2,p_1\big)\;,\qquad\quad\forall\ (p_1,p_2)\in\bb{E}\times_\pi\bb{E}
$$
where $\bb{E}\times_\pi\bb{E}
:=\{(p_1,p_2)\in\bb{E}\times\bb{E}\ |\ \pi(p_1)=\pi(p_2)\}$.
 
  \medskip
 
A vector bundle \emph{morphism} $f$ between two vector bundles  $\pi:\bb{E}\to X$ and $\pi':\bb{E}'\to X$ 
over the same base space
is a  continuous map $f:\bb{E}\to \bb{E}'$ which is \emph{fiber preserving} in the sense that  $\pi=\pi'\circ f$
 and that restricts to a \emph{linear} map on each fiber $\left.f\right|_x:\bb{E}_x\to \bb{E}'_x$. Complex vector bundles over  $X$ together with vector bundle morphisms define a category and the symbol ${\rm Vec}^m_\C(X)$
 is used to denote the set of equivalence classes of isomorphic vector bundles of rank $m$.
    Also 
Q-bundles define a category with respect to  \emph{Q-morphisms}. A Q-morphism $f$ between two  Q-bundles
 $(\bb{E},\Theta)$ and $(\bb{E}',\Theta')$ over the same involutive space $(X,\tau)$ 
  is a vector bundle morphism  commuting with the involutions, \ie $f\circ\Theta\;=\;\Theta'\circ f$. The set of equivalence classes of isomorphic Q-bundles of  rank $m$ over $(X,\tau)$ 
  will be denoted by ${\rm Vec}_{Q}^m(X,\tau)$.

\begin{remark}[{\virg{Real}} vector bundles]\label{rk_real}{\upshape
By changing condition  $(Q_3)$
 in Definition \ref{defi:Q_VB}   with 
\begin{itemize}
\item[$(R)$] \emph{$\Theta^2$ acts fiberwise as the multiplication by $1$, namely $\Theta^2|_{\bb{E}_x}=\n{1}_{\bb{E}_x}$}\;
\end{itemize}
one ends in the category of \emph{\virg{Real}} (or R-)  \emph{vector bundles}. 
The set of isomorphism classes of rank $m$ R-bundles  over the involutive space $(X,\tau)$ 
is denoted by ${\rm Vec}_{R}^m(X,\tau)$. 
For more details we refer to \cite{atiyah-66,denittis-gomi-14}.
}\hfill $\blacktriangleleft$
\end{remark}

\medskip

In the case of a trivial involutive space $(X,{\rm Id}_X)$ one has the isomorphisms
\begin{equation}\label{eq:iso_trivial_involution}
{\rm Vec}_{Q}^{2m}(X,{\rm Id}_X)\,\simeq\,{\rm Vec}_{\n{H}}^{m}(X)\;,\qquad\quad {\rm Vec}_{R}^{m}(X,{\rm Id}_X)\,\simeq\,{\rm Vec}_{\n{R}}^{m}(X,{\rm Id}_X)\;,\qquad \ \ m\in\N
\end{equation}
where ${\rm Vec}_{\n{F}}^{m}(X)$ is the set of equivalence classes of vector bundles over
$X$ with typical fiber $\n{F}^m$ and $\n{H}$ denotes the skew field of quaternions. The first isomorphism in \eqref{eq:iso_trivial_involution} is proved in \cite{dupont-69} (see also \cite[Proposition 2.2]{denittis-gomi-14-gen}) while the proof of the second is provided in \cite{atiyah-66} (see also \cite[Proposition 4.5]{denittis-gomi-14}). These two results justify the names \virg{Quaternionic} and \virg{Real} for the related categories.

\medskip

 Let $x\in X^\tau$ and $\bb{E}_x\simeq\C^m$ be the related fiber. In this case the restriction $\Theta|_{\bb{E}_x}\equiv J$ defines an  \emph{anti}-linear map $J : \bb{E}_x \to \bb{E}_x$ such that $J^2 = -\n{1}_{\bb{E}_x}$. Said differently the fibers $\bb{E}_x$ over fixed points $x\in X^\tau$ are endowed with a \emph{quaternionic} structure (\cf \cite[Remark 2.1]{denittis-gomi-14-gen}).
This fact has an important  consequence (\cf \cite[Proposition 2.1]{denittis-gomi-14-gen}): 
\begin{proposition}
If $X^\tau\neq \emptyset$ then every Q-bundle over $(X,\tau)$   has necessarily even rank.
\end{proposition}

\medskip

\noindent
The set
 ${\rm Vec}_{Q}^{2m}(X,\tau)$ is non-empty since it contains at least the 
 \emph{trivial} element in the \virg{Quaternionic} category.
\begin{definition}[\virg{Quaternionic} product  bundle]
The rank $2m$  product Q-bundle over the involutive space $(X,\tau)$ is the complex vector bundle
$$
X\,\times\,\C^{2m}\,\longrightarrow\, X
$$
endowed with the \emph{product} $Q$-structure
$$
\Theta_0 (x,{\rm v})\,=\,(\tau(x),Q\;\overline{{\rm v}})\;,\qquad\quad (x,{\rm v})\,\in\,X\,\times\,\C^{2m}
$$ 
where the matrix $Q$ is given by
\beql{eq:Q-mat}
Q
\;:=\;\left(
\begin{array}{rr}
0 & -1    \\
1 &  0 
\end{array}
\right)\;\otimes\;\n{1}_m\;=\;
\left(
\begin{array}{rr|rr|rr}
0 & -1 &        &        &   &    \\
1 &  0 &        &        &   &    \\
\hline
  &    & \ddots &        &   &    \\
  &    &        & \ddots &   &    \\
\hline
  &    &        &        & 0 & -1 \\
  &    &        &        & 1 &  0
\end{array}
\right)\;.
\eeq
A \virg{Quaternionic} vector bundle is called \emph{Q-trivial}  if it is isomorphic  to the  product  Q-bundle.
\end{definition}

\begin{remark}[Odd rank case]\label{rk:odd_rank}{\upshape
Let us point out that in the case of a free involution $X^\tau= \emptyset$   the sets ${\rm Vec}_{Q}^{2m-1}(X,\tau)$ can be non-empty but in general there is no obvious candidate for the 
role of the trivial element. ``Quaternionic'' line bundles ($m=1$) have been studied and classified in \cite[Section 3]{denittis-gomi-16}. The classification of Q-bundles of odd rank in low dimension is provided 
in \cite[Theorem 1.2]{denittis-gomi-18}.
Anyway, in this work we will not be interested in the case of free involutions and therefore we will not consider odd rank Q-bundles.
}\hfill $\blacktriangleleft$
\end{remark}

A \emph{section} of a complex vector bundle $\pi:\bb{E}\to X$ is a continuous map $s:X\to\bb{E}$ such that $\pi\circ s ={\rm Id}_X$. The set $\Gamma(\bb{E})$ has the structure of a left $C(X)$-module with multiplication given by the pointwise product $(fs)(x):=f(x)s(x)$ for any $f\in C(X)$ and $s\in\Gamma(\bb{E})$ and for all $x\in X$. 
If $(\bb{E}, \Theta)$ is a Q-bundle over $(X,\tau)$ then $\Gamma(\bb{E})$ is endowed with a natural anti-linear {anti-involution} $\tau_\Theta:\Gamma(\bb{E})\to\Gamma(\bb{E})$ given by
$$
\tau_\Theta(s)\,:=\,\Theta\,\circ\, s\,\circ\, \tau\,.
$$
The compatibility with the $C(X)$-module structure is given by
$$
\tau_\Theta(fs)\,=\,\tau_*(f)\,\tau_\Theta(s)\,,\qquad\quad f\in C(X)\,,\ \ s\in\Gamma(\bb{E})
$$
where the anti-linear {involution} $\tau_*:C(X)\to C(X)$ is defined by $\tau_*(f)(x):=\overline{f(\tau(x))}$. The triviality of a \virg{Quaternionic} vector bundle can be characterized in terms of global Q-frames of sections \cite[Definition 2.1 \& Theorem 2.1]{denittis-gomi-14-gen}.

  %----------------------------%
\subsection{Stable range in low dimension}
\label{sect:stable_range}
The {stable rank condition} for vector bundles expresses the pretty general fact that the non trivial topology can be concentrated in a sub-vector bundle of \emph{minimal} rank. This minimal value  depends on the dimensionality of the base space and on the category of vector bundles under consideration. For complex (as well as real or quaternionic) vector bundles the  stable rank condition is a well-known result (see \eg \cite[Chapter 9, Theorem 1.2]{husemoller-94}).
The proof of the latter is based on an \virg{obstruction-type argument} which provides the  construction of a certain \emph{maximal} number of global sections \cite[Chapter 2, Theorem 7.1]{husemoller-94}. 

\medskip

The  latter argument can be generalized to vector bundles over spaces with involution 
by means of the notion of  $\Z_2$-CW-complex \cite{matumoto-71,allday-puppe-93} (see also \cite[Section 4.5]{denittis-gomi-14}). A $\Z_2$-CW-complex is a  CW-complex  made by cells of various dimension that carry a $\Z_2$-action. 
These $\Z_2$-cells can be only of two types: They are \emph{fixed} if the action of $\Z_2$ is trivial or they are \emph{free} if they have no fixed points. 
Since this construction is modeled after the usual definition of CW-complex, just by replacing the \virg{point} by \virg{$\Z_2$-point}, (almost) all topological and homological properties valid for  CW-complexes have their {natural} counterparts in the equivariant setting. The use of this technique is essential		for the determination of  the 
stable rank condition in the case of R-bundles  \cite[Theorem 4.25]{denittis-gomi-14} and   Q-bundles  \cite[Theorem 4.2 \& Theorem 4.5]{denittis-gomi-16}.

\medskip

In this section we recall 
 the results about the stable range for R-bundles and (even rank) Q-bundles over low dimensional base spaces. Indeed, these are the only   cases of interest of the present work.
\begin{theorem}[Stable condition in low dimension]
\label{theo:stab_ran_Q_even}
Let $(X,\tau)$ be an involutive space such that $X$ has a finite $\Z_2$-CW-complex decomposition of dimension $d$. Assume that $X^\tau\neq\emptyset$ is a $\Z_2$-CW-complex of dimension zero.
Then it holds true that:
\begin{itemize}
\item[] \emph{- Stable condition for $R$-bundles -}
$$
\begin{aligned}
{\rm Vec}^{m}_{R}(X, \tau)&\;=\; 0&\quad \text{if}& \ \ d=0,1\\
{\rm Vec}^{m}_{R}(X, \tau)&\;\simeq\; {\rm Vec}^{1}_{R}(X, \tau)& \quad\text{if}& \ \ 2\leqslant d\leqslant 3
\end{aligned}\qquad \forall\ m\in\N\;;
$$
\vspace{1mm}
\item[] \emph{- Stable condition for $Q$-bundles -}
$$
\begin{aligned}
{\rm Vec}^{2m}_{Q}(X, \tau)&\;=\; 0&\quad \text{if}& \ \ d=0,1\\
{\rm Vec}^{2m}_{Q}(X, \tau)&\;\simeq\; {\rm Vec}^{2}_{Q}(X, \tau)& \quad\text{if}& \ \ 2\leqslant d\leqslant 5
\end{aligned}\qquad \forall\ m\in\N\;.
$$\end{itemize}
\end{theorem}
\noindent 
In particular, under the hypothesis of validity of Theorem \ref{theo:stab_ran_Q_even} the dimensions $d=0,1$ are trivial since in these cases only  the trivial R- and Q-bundles (up to isomorphism) exist.
In the cases $d=2,3$, which are the really interesting cases for this work, it is enough to study the sets 
${\rm Vec}^{1}_{R}(X, \tau)$ and ${\rm Vec}^{2}_{Q}(X, \tau)$.

  %----------------------------%
\subsection{The FKMM-invariant}
\label{subsec:fkmm-invariant}
 Q-bundles  can be classified, at least partially, by means of a characteristic class called \emph{{FKMM}-invariant}. This topological object has been firstly introduced in \cite{furuta-kametani-matsue-minami-00}
 and then studied  and generalized in  \cite{denittis-gomi-14-gen,denittis-gomi-16,denittis-gomi-18}. In this section we review the main properties of the {FKMM}-invariant.

\medskip

Let $(X,\tau)$ be an involutive space and $X^\tau\subseteq X$ its fixed point subset. In order to introduce the  {FKMM}-invariant one needs the 
 \emph{equivariant Borel cohomology} ring 
of $(X,\tau)$
with coefficients
in the local systems $\Z(1)$; \ie
\begin{equation}\label{eq:homot_quot}
H^\bullet_{\Z_2}(X,\Z(1))\;:=\; H^\bullet({X}_{\sim\tau},\Z(1))\;.
\end{equation}
More precisely, each equivariant cohomology group $H^j_{\Z_2}(X,\Z(1))$ is given by the
 singular cohomology group  $H^j({X}_{\sim\tau},\Z(1))$ of the  homotopy quotient 
$$
{X}_{\sim\tau}\;:=\;X\times\ {\n{S}}^{0,\infty} /( \tau\times \theta_\infty)\;.
$$
 where $\theta_\infty$ is the {antipodal map} on the infinite sphere $\n{S}^\infty$ .
The local system  $\Z(1)$ over $(X,\tau)$ can be identified with the product space $\Z(1)\simeq X\times\Z$ made equivariant  by the $\Z_2$-action $(x,l)\mapsto(\tau(x),-l)$.
The fixed point subset $X^\tau$ is closed in $X$ and $\tau$-invariant. The inclusion $\imath:X^\tau\hookrightarrow X$ extends to an inclusion $\imath:X^\tau_{\sim\tau}\hookrightarrow X_{\sim\tau}$ of the respective homotopy quotients. The \emph{relative} equivariant cohomology can be defined as usual by the identification
$$
H^\bullet_{\Z_2}(X|X^\tau,\Z(1))\;:=\; H^\bullet({X}_{\sim\tau}|X^\tau_{\sim\tau},\Z(1))\;.
$$
For a more detailed description of the equivariant Borel cohomology we refer to \cite[Section 3.1]{denittis-gomi-14-gen} and references therein.

 \medskip
 
 The {FKMM}-invariant is a semi-group homomorphism
\begin{equation}\label{eq:FKMM_inv_map}
{\rm Vec}_{Q}^{2m}(X,\tau)\ni\;[(\bb{E},\Theta)]\;\stackrel{\kappa}{\longrightarrow}\;\kappa(\bb{E},\Theta)\;\in\;H^2_{\Z_2}(X|X^\tau,\Z(1))\;
\end{equation}
which associates to the isomorphism class  $[(\bb{E},\Theta)]$ of the Q-bundle $(\bb{E},\Theta)$
 a cohomology class $\kappa(\bb{E},\Theta)$ in the {relative} equivariant cohomology group 
 $H^2_{\Z_2}(X|X^\tau,\Z(1))$. The semi-group structure in ${\rm Vec}_{Q}^{2m}(X,\tau)$ is given by the Whitney sum. The construction of the map $\kappa$ has been firstly described in \cite[Section 3.3]{denittis-gomi-14-gen} and then generalized in \cite[Section 2.5]{denittis-gomi-16}. In this section
 we will skip the details of the construction of the {FKMM}-invariant  while we will focus on its relevant properties of
the  map \eqref{eq:FKMM_inv_map}:
\begin{itemize}
\item[(a)] Isomorphic Q-bundles define the same FKMM-invariant;
\vspace{1mm}
\item[(b)] The FKMM-invariant is \emph{natural} under the pullback induced by equivariant maps;
\vspace{1mm}
\item[(c)] If $(\bb{E},\Theta)$ is Q-trivial then $\kappa(\bb{E},\Theta)=0$; 
\vspace{1mm}
\item[(d)] The FKMM-invariant is {additive} with respect to the Whitney sum and the abelian structure of $H^2_{\Z_2}(X|X^\tau,\Z(1))$. More precisely
$$
\kappa(\bb{E}_1\oplus \bb{E}_2,\Theta_1\oplus\Theta_2)\;=\;\kappa(\bb{E}_1,\Theta_1)\;\cdot\;\kappa( \bb{E}_2,\Theta_2)
$$
for each pair of Q-bundles $(\bb{E}_1,\Theta_1)$ and $(\bb{E}_2,\Theta_2)$ over the same involutive space $(X,\tau)$.
\end{itemize}
For the justification of these properties we refer to \cite[Section 2.6]{denittis-gomi-16}.

  %----------------------------%
\subsection{Topological classification over low-dimensional FKMM-spaces}
\label{subsec:top_class_fkmm_space}
 The FKMM-invariant is an extremely efficient tool for the classification of Q-bundles in low dimension. 
 The first observation is that in  great generality the FKMM-invariant is injective in low dimension, \ie when the base space has dimension $0\leqslant d\leqslant 3$. More precisely, as a consequence of \cite[Theorem 4.7 \& Theorem 4.9]{denittis-gomi-16} one has that: 
 \begin{theorem}[Injectivity in low dimension]
 Let $(X,\tau)$ be an involutive space which verifies Assumption \ref{ass:top}. Let its dimension be $d = 0, 1,2,3$.
 Then the map \eqref{eq:FKMM_inv_map} is injective.
  \end{theorem}
 \noindent
 This result suggests that in low dimension the invariant $\kappa$ can be used to label inequivalent classes of Q-bundles by means of elements of the cohomology group $H^2_{\Z_2}(X|X^\tau,\Z(1))$. The next natural questions is about the surjectivity of the map $\kappa$. In this case is possible to provide a general positive answer only if $0\leqslant d\leqslant 2$. As proved in \cite[Corollary 4.2 \& Proposition 4.9]{denittis-gomi-18} one has that:
 \begin{theorem}[Surjectivity in dimension two]
 \label{theo:biject_K}
 Let $(X,\tau)$ be an involutive space of dimension $d=2$ which verifies Assumption \ref{ass:top}. 
 Then 
 $$
 {\rm Vec}^{2m}_{Q}(X, \tau)\;\simeq\; H^2_{\Z_2}(X|X^\tau,\Z(1))\;\qquad \forall\ m\in\N\;,
$$
namely the map \eqref{eq:FKMM_inv_map} is bijective.
 \end{theorem}
 \noindent
 Theorem \ref{theo:biject_K} can be juxtaposed
 with  the stable condition described in Theorem \ref{theo:stab_ran_Q_even}
$$
{\rm Vec}^{2m}_{Q}(X, \tau)\;=\; 0\;\qquad\quad \text{if}\ \ d=0,1
\qquad \forall\ m\in\N\;
$$
to obtain a  complete classification of Q-bundles in dimension $d=0,1,2$.

 \medskip
 
In the case $d=3$ the surjectivity of the FKMM-invariant generally fails as shown by the example presented in \cite[Section 5]{denittis-gomi-18}. However the surjectivity can be recovered by requiring  some extra property to the base space $(X,\tau)$. In the next part of this work we will be mainly focused on spaces of the following type:
\begin{definition}[FKMM-manifold]
\label{def_FKMM-space}
An involutive space $(X,\tau)$ is called FKMM-manifold if:
\begin{itemize}
\item[(a)] $X$ is a compact Hausdorff manifold without boundary;
\vspace{1mm}
\item[(b)] The involution $\tau$ preserves the manifold structure;
\vspace{1mm} 
\item[(c)] The fixed point set $X^\tau$ consists at most
 of a finite collection of points;
\vspace{1mm} 
\item[(d)] $H^2_{\Z_2}(X,\Z(1))=0$.
\vspace{1mm} 
\end{itemize}
\end{definition}
\noindent 
 Let us observe that an involutive space  $(X,\tau)$ which fulfills conditions (a) and (b) in Definition \ref{def_FKMM-space}  is a \emph{closed manifold} which automatically admits the structure of a $\Z_2$-CW-complex (see, \eg  \cite[Theorem 3.6]{may-96}). Then an   FKMM-manifold meets all the requirements stated in Assumption \ref{ass:top}. The conditions (c) and (d) are the crucial ingredients for the definition of a \emph{topological}  FKMM-space according to the original definition \cite[Definition 1.1]{denittis-gomi-14-gen}. The requirement of a manifold structure has a twofold justification: First of all it allows the use of a technical tool (the slice theorem) in the proof of the crucial result \cite[Proposition 4.13]{denittis-gomi-18};
Secondly, the main aim of this work is the study of the classification of Q-bundles over involutive  manifolds (see Section \ref{sect:geom_class}). 
The manifold structure and the map $\tau$ are tacitly assumed to be of some given regularity (\eg. $C^r$ or smooth).
The next result provides the topological classification of  Q-bundles over 
low dimensional FKMM-manifolds.
 \begin{theorem}[Classification for  FKMM-manifolds]
 \label{theo:biject_K_FKMM}
 Let $(X,\tau)$ be an  FKMM-manifold of dimension $0\leqslant d\leqslant3$. Then it holds true that
 $$
\begin{aligned}
{\rm Vec}^{2m}_{Q}(X, \tau)&\;=\; 0&\quad \text{if}& \ \ d=0,1\\
{\rm Vec}^{2m}_{Q}(X, \tau)&\;\simeq\; H^2_{\Z_2}(X|X^\tau,\Z(1))& \quad\text{if}& \ \ d=2,3
\end{aligned}\qquad \forall\ m\in\N\;
$$
and the isomorphism (in the non-trivial cases) is given by the FKMM-invariant \eqref{eq:FKMM_inv_map}.
 \end{theorem}
 \noindent 
 The cases $d=0,1$ are a consequence of the  stable condition described in Theorem \ref{theo:stab_ran_Q_even}. The case $d=2$ follows from Theorem \ref{theo:biject_K}.
 Finally the new case  $d=3$ is proved in \cite[Proposition 4.13]{denittis-gomi-18}.

\medskip

 Let us observe that Theorem \ref{theo:biject_K_FKMM} trivially holds also for the free involution case $X^\tau=\emptyset$. In this case, as a consequence of the condition (d) in Definition \ref{def_FKMM-space} one has that $H^2_{\Z_2}(X|\emptyset,\Z(1))\simeq H^2_{\Z_2}(X,\Z(1))=0$
 which implies that an FKMM-manifold with free involution only supports the trivial Q-bundle.
 In order to focus on the non-trivial situations we will assume henceforth that $d=2,3$ and 
$X^\tau\neq\emptyset$.

\medskip

When  $(X,\tau)$ is an FKMM-manifold then the cohomology group $H^2_{\Z_2}(X|X^\tau,\Z(1))$ has a 
explicit representation in terms of (equivalence classes of) maps. As proved in \cite[Lemma 3.1]{denittis-gomi-14-gen} one has the 
 following isomorphism 
\begin{equation}\label{eq:iso_for_fkmm}
H^2_{\Z_2}(X|X^\tau,\Z(1))\;\simeq\;{\rm Map}(X^\tau,\{\pm1\})\;/\;[X,\n{U}(1)]_{\Z_2}
\end{equation}
where ${\rm Map}(X^\tau,\{\pm1\})\simeq{\Z_2}^{|X^\tau|}$ is the set of maps from $X^\tau$ to $\{\pm1\}$ (recall  that $X^\tau$ is a set of finitely many points) and
$[X,\n{U}(1)]_{\Z_2}$ denotes the set of classes of $\Z_2$-homotopy equivalent equivariant maps between the involutive space $(X,\tau)$ and the group $\n{U}(1)$ endowed with the involution given by the complex conjugation. 
The group action of $[X,\n{U}(1)]_{\Z_2}$ on ${\rm Map}(X^\tau,\{\pm1\})$ is given by multiplication
and restriction. More precisely, let 
$[u]\in [X,\n{U}(1)]_{\Z_2}$ and $s\in {\rm Map}(X^\tau,\{\pm1\})$, then the action of $[u]$ on $s$ is given by $s\mapsto [u](s):=u|_{X^\tau} s$. By combining Theorem \ref{theo:biject_K_FKMM} with the isomorphism \eqref{eq:iso_for_fkmm} one gets the following result:
\begin{proposition}\label{prop:fkmm-inv_fkmm-space}
Let $(X,\tau)$ be an FKMM-manifold of dimension $d=2,3$ and assume that $X^\tau\neq\emptyset$.
  Then, the FKMM-invariant $\kappa$ induces an isomorphisms
$$
\begin{aligned}
{\rm Vec}^{2m}_{Q}(X, \tau)&\;\simeq\; {\rm Map}(X^\tau,\{\pm1\})\;/\;[X,\n{U}(1)]_{\Z_2}& \quad&
\end{aligned}\qquad \forall\ m\in\N\;.
$$
\end{proposition}

\medskip

In summary the content of Theorem \ref{theo:biject_K_FKMM} and {Proposition} \ref{prop:fkmm-inv_fkmm-space} is the following:
Every Q-bundles $(\bb{E},\Theta)$ over an FKMM-space $(X,\tau)$ of dimension $d=2,3$ such that $X^\tau\neq\emptyset$
is classified by its FKMM-invariant $\kappa(\bb{E},\Theta)$. The latter can be represented as a map
$$
s_{(\bb{E},\Theta)}\;:\; X^\tau\;\longrightarrow\;\{\pm 1\}
$$
modulo the (right) multiplication by the restriction over $X^\tau$ of an equivariant function $u:X\to\n{U}(1)$. The map $s_{(\bb{E},\Theta)}$ is called the \emph{canonical section} associated to $(\bb{E},\Theta)$ and its construction is described in \cite[Section 3.2]{denittis-gomi-14-gen} or \cite[Section 2.2]{denittis-gomi-16}.

  %----------------------------%
\subsection{The Fu-Kane-Mele index}
\label{subsec:fu-kane-mele-invariant}
Let us focus on the non-trivial case of an   FKMM-manifold $(X,\tau)$ (see Definition \ref{def_FKMM-space}) of dimension $d=2,3$ such that $X^\tau\neq\emptyset$. At the end of Section \ref{subsec:top_class_fkmm_space} we showed that every Q-bundle $(\bb{E},\Theta)$ over $(X,\tau)$ is classified by a map  $s_{(\bb{E},\Theta)}\in {\rm Map}(X^\tau,\{\pm1\})$, called the canonical section,
 modulo the action (multiplication and restriction) of an equivariant map $u:X\to\n{U}(1)$.
 Clearly $(\bb{E},\Theta)$ is equivalently classified by any other map $\phi\in {\rm Map}(X^\tau,\{\pm1\})$ in the same equivalence class of $s_{(\bb{E},\Theta)}$, namely by any representative of 
  $[s_{(\bb{E},\Theta)}]\in {\rm Map}(X^\tau,\{\pm1\})/[X,\n{U}(1)]_{\Z_2}$.

 \medskip
 
Consider now the  \emph{product sign} map
\begin{equation}\label{eq:product_index_map}
\Pi:{\rm Map}\big(X^\tau,\{\pm 1\}\big)\;\longrightarrow\; \{\pm 1\}
\end{equation}
defined by
\begin{equation}\label{eq:Pi_map}
\Pi(\phi)\;:=\;\prod_{j=1}^{|X^\tau|}\phi(x_j)\;\qquad\quad \phi\;\in\; {\rm Map}(X^\tau,\{\pm 1\})\;.
\end{equation}
The value $\Pi(\phi)$ is called the \emph{Fu-Kane-Mele index} of $\phi$.
 There is no  reason to suspect \emph{a priori} that the Fu-Kane-Mele index is well defined on the equivalence classes in ${\rm Map}(X^\tau,\{\pm 1\})/[X,\n{U}(1)]_{\Z_2}$. In fact, if $\phi_1$ and $\phi_2$ were two representatives of the same class  $[\phi]\in{\rm Map}(X^\tau,\{\pm 1\})/[X,\n{U}(1)]_{\Z_2}$ related by an equivariant function $u:X\to\n{U}(1)$ which takes an odd number of time the value $-1$ on  $X^\tau$ one would have that $\Pi[\phi_1]=-\Pi[\phi_2]$. For this reason the following result, proved in \cite[Proposition 4.5 \& Theorem 4.2]{denittis-gomi-14-gen} is, at  first glance, quite surprising from a topological point of view.
\begin{theorem}[Fu-Kane-Mele formula, $d=2$]
\label{prop:Fu-Kane-Mele_formula1}
Let $(X,\tau)$ be an oriented two-dimensional FKMM-manifold in the sense of Definition \ref{def:good_manif_d=2}.
Then, $(X,\tau)$ is 
an FKMM-manifold according to  Definition \ref{def_FKMM-space}.
Moreover
\begin{equation}\label{iso_manithing}
H^2_{\Z_2}\big(X|X^\tau,\Z(1)\big)\;\simeq\;\Z_2\;\simeq\;\{\pm 1\}
\end{equation}
where in the second isomorphism  the cyclic group $\Z_2$ is identified with the multiplicative group $\{\pm 1\}$. Moreover, any Q-bundle $(\bb{E},\Theta)$ over $(X,\tau)$ is classified by the FKMM-invariant $\kappa(\bb{E},\Theta)\in \{\pm 1\}$ which can be computed by 
$$
\kappa(\bb{E},\Theta)\;=\;\Pi(\phi) 
$$
where $\Pi$ is the {product sign} map \eqref{eq:product_index_map} and $\phi\in[s_{\bb{E}}]$
is any representative of the class $[s_{(\bb{E},\Theta)}]\in {\rm Map}(X^\tau,\{\pm1\})/[X,\n{U}(1)]_{\Z_2}$ of the canonical section.
\end{theorem}
\proof[{Proof} (sketch of)]
Clearly conditions (a'), (b') and (c') of Definition \ref{def:good_manif_d=2} imply 
conditions (a), (b) and (c) of Definition \ref{def_FKMM-space}. Moreover, \cite[Proposition 4.4]{denittis-gomi-14-gen} assures that (a'), (b') and (c') imply condition (d) of Definition \ref{def_FKMM-space}, \ie  $H^2_{\Z_2}(X,\Z(1))=0$ along with isomorphism \eqref{iso_manithing}. The rest of the claim is proved in \cite[Proposition 4.5 \& Theorem 4.2]{denittis-gomi-14-gen}.
\qed
\medskip

\noindent 
 As a byproduct of Theorem \ref{prop:Fu-Kane-Mele_formula1} one has that the {Fu-Kane-Mele index} is unambiguously defined on the   whole equivalence class $[s_{(\bb{E},\Theta)}]$ and the 
 Q-bundle $(\bb{E},\Theta)$ is classified, up to isomorphisms, by the sign $\Pi(\phi) \in\{\pm 1\}$ 
 where $\phi\in  {\rm Map}(X^\tau,\{\pm1\})$ is any map which differs from $s_{(\bb{E},\Theta)}$
 by the multiplication (and restriction)  by an equivariant map $u:X\to\n{U}(1)$.

\medskip

Although with some differences,
the next result pairs Theorem  \ref{prop:Fu-Kane-Mele_formula1}
 in dimension $d=3$. It
 can be considered one of the main achievements of this work.
\begin{theorem}[Fu-Kane-Mele formula, $d=3$]
\label{prop:Fu-Kane-Mele_formula2}
Let $(X,\tau)$ be an FKMM-manifold (see Definition \ref{def_FKMM-space}) of dimension $d=3$ with $X^\tau\neq\emptyset$.
Assume in addition that:
\begin{itemize}
\item[(e)] $X$  is oriented and $\tau$ \emph{reverses} the orientation.
\end{itemize}
Let  $(\bb{E},\Theta)$ be a Q-bundle over $(X,\tau)$ with FKMM-invariant $\kappa(\bb{E},\Theta)$ represented by the class $[s_{(\bb{E},\Theta)}]\in {\rm Map}(X^\tau,\{\pm 1\})/[X,\n{U}(1)]_{\Z_2}$ according to Proposition \ref{prop:fkmm-inv_fkmm-space}. Then, the sign
\begin{equation}\label{eq:strong_FKMM}
\kappa_{\rm s}(\bb{E},\Theta)\;:=\;\Pi[\phi] 
\end{equation}
is independent of the  choice of the representative $\phi\in [s_{(\bb{E},\Theta)}]$ and provides a topological invariant for $(\bb{E},\Theta)$.
\end{theorem}
\noindent 
Theorem \ref{prop:Fu-Kane-Mele_formula2} follows as  a consequence of Theorem \ref{prop:Fu-Kane-Mele=chern_d=3} which will be proved in Section \ref{sect:class_d=3}. It is worth noting that, even though Theorem \ref{prop:Fu-Kane-Mele_formula1} and Theorem \ref{prop:Fu-Kane-Mele_formula2} seem to be of topological nature, 
they need the manifold structure of $X$. In particular  Theorem \ref{prop:Fu-Kane-Mele=chern_d=3}, which implies Theorem \ref{prop:Fu-Kane-Mele_formula2}, relies on differential geometric techniques. 

\medskip

The quantity $\kappa_{\rm s}(\bb{E},\Theta)$ in Proposition \ref{prop:Fu-Kane-Mele_formula2} in general does not specify completely the FKMM-invariant of $(\bb{E},\Theta)$, but only a part of it. We refer to $\kappa_{\rm s}(\bb{E},\Theta)$ as the \emph{strong component} of the FKMM-invariant.

  %----------------------------%
\subsection{Alternative presentation of \virg{Quaternionic} vector bundles in low-dimensions}
\label{subsec:altern_present}

This section is focused on an alternative  description of rank 2 Q-bundles over low-dimensional involutive spaces $(X,\tau)$ such that  
$H^2_{\Z_2}(X,\Z(1))=0$. It is worth mentioning that under these conditions the complex vector bundle underlying each Q-bundle  is necessarily trivial \cite[Proposition 4.1]{denittis-gomi-14-gen}.

\medskip

Let ${\rm Map}(X,\n{SU}(2))$ be the space of (smooth) maps from $X$ into $\n{SU}(2)$. Given $\xi\in{\rm Map}(X,\n{SU}(2))$
let $\tau^*\xi$ be the map defined by $\tau^*\xi(x):=\xi(\tau(x))$ for all $x\in X$. The space of equivariant maps 
 from $X$ into $\n{SU}(2)$ is defined by
\begin{equation}\label{eq:coxxi-01}
 {\rm Map}(X,\n{SU}(2))_{\Z_2}\;:=\;\big\{\xi\in {\rm Map}(X,\n{SU}(2))\ |\  \tau^*\xi=\xi^{-1}\big\}\;.
\end{equation}
 The set of classes of $\Z_2$-homotopy equivalent maps, denoted with $[X, \n{SU}(2)]_{\Z_2}$, inherits a group structure from ${\rm Map}(X,\n{SU}(2))_{\Z_2}$. Let us consider also the groups
\begin{equation}\label{eq:coxxi-02}
 \begin{aligned}
 {\rm Map}(X,\n{U}(2))'_{\Z_2}\;&:=\;\big\{\psi\in {\rm Map}(X,\n{U}(2))\ |\  {\rm det}(\tau^*\psi)={\rm det}(\overline{\psi})\big\}\;\\
 {\rm Map}(X,\n{U}(1))_{\Z_2}\;&:=\;\big\{\phi\in {\rm Map}(X,\n{U}(1))\ |\  \tau^*\phi=\overline{\phi}\big\}\;
 \end{aligned}
 \end{equation}
where $\overline{\psi}$ and $\overline{\phi}$ are the complex conjugated of ${\psi}$ and ${\phi}$, respectively. 
  The related sets of equivalence classes under the $\Z_2$-homotopy are denoted with   
$[X,\n{U}(2)]'_{\Z_2}$ and $[X,\n{U}(1)]_{\Z_2}$, respectively.

\medskip

By construction one has the inclusion ${\rm Map}(X,\n{SU}(2))_{\Z_2}\subset {\rm Map}(X,\n{U}(2))'_{\Z_2}$. Moreover the group ${\rm Map}(X,\n{U}(2))'_{\Z_2}$ acts on 
${\rm Map}(X,\n{SU}(2))_{\Z_2}$ on the following way: Given $\psi\in {\rm Map}(X,\n{U}(2))'_{\Z_2}$ let $G_\psi$ be the automorphism of ${\rm Map}(X,\n{SU}(2))_{\Z_2}$ given by 
\begin{equation}\label{eq:def_G}
G_\psi(\xi)\;:=\;-(\tau^*\psi^{-1}) \xi Q \overline{\psi} Q\;,\qquad\quad \xi\in{\rm Map}(X,\n{SU}(2))_{\Z_2}
\end{equation}
where the dot $\cdot$ denotes the matrix multiplication  and $Q$ is the (size $2\times2$) matrix \eqref{eq:Q-mat}. In fact, given that ${\rm det}(\tau^*\psi^{-1})={{\rm det}(\tau^*\psi)}^{-1}={\rm det}(\overline{\psi})^{-1}$ it follows that
 ${\rm det}(G_\psi(\xi))={\rm det}(\xi)=1$.
 Moreover, the equality $\tau^*G_\psi(\xi)=G_\psi(\xi)^{-1}$ follows from a direct calculation along with the equality $Q \xi=\overline{\xi}  Q$
valid for maps with values in  $\n{SU}(2)$.

\medskip

The main aim of this section is to prove the following result:
\begin{theorem}\label{theo:main-SU2_charat}
Let $(X,\tau)$ be  an
 involutive space of dimension $0\leqslant d\leqslant 2$ which meets  Assumption \ref{ass:top}. Assume in addition that $H^2_{\Z_2}(X,\Z(1))=0$ in the case $d=2$. Then there is a natural group isomorphism
$$
{\rm Vec}^{2}_{Q}(X, \tau)\;\simeq\; [X,\n{SU}(2)]_{\Z_2}\;/\;[X,\n{U}(1)]_{\Z_2}
$$
where the action of $[X,\n{U}(1)]_{\Z_2}$ on $[X,\n{SU}(2)]_{\Z_2}$ is defined as follows: Given $[\phi]\in [X,\n{U}(1)]_{\Z_2}$ let $L_{[\phi]}$ be the automorphism of $[X,\n{SU}(2)]_{\Z_2}$ defined by
$$
L_{[\phi]}([\xi])\;:=\;\left[\left(
\begin{array}{cc}
\tau^*\phi & 0 \\
0 & 1
\end{array}
\right)\; \xi\;\left(
\begin{array}{cc}
1& 0 \\
0 & \phi
\end{array}
\right)\right]\;.
$$
\end{theorem}

We start with  a couple of preliminary results which are  valid in dimension $0\leqslant d\leqslant 3$.
\begin{lemma}
\label{lemma:prep1}
Let $(X,\tau)$ be a  low-dimensional 
 involutive space which meets  Assumption \ref{ass:top}. Assume in addition that $H^2_{\Z_2}(X,\Z(1))=0$ in the case $d=2,3$.
Then, there is a natural bijection
\begin{equation}\label{eq:first_biject}
{\rm Vec}^{2}_{Q}(X, \tau)\;\simeq\; {\rm Map}(X,\n{SU}(2))_{\Z_2}\;/\;{\rm Map}(X,\n{U}(2))'_{\Z_2}
\end{equation}
where the action of ${\rm Map}(X,\n{U}(2))'_{\Z_2}$ on ${\rm Map}(X,\n{SU}(2))_{\Z_2}$ is given by the automorphisms \eqref{eq:def_G}.
\end{lemma}
\begin{proof}
Let $\pi : \bb{E} \to X$ be a rank 2 $Q$-bundle. The low dimensionality of the base space  implies that the underlying complex vector bundle $\bb{E}$ is isomorphic to the product bundle $X \times \C^2$
\cite[Proposition 4.1]{denittis-gomi-14-gen}. The induced Q-structure $\Theta$ on $X \times \C^2$  is then expressed through
 a function $\xi:X\to\n{U}(2)$ in the form 
$\Theta:(x, v) \mapsto (\tau(x), \xi(x) Q\overline{v})$ and the \virg{Quaternionic} condition is guaranteed by the constraint $
\tau^*\xi=-Q  \overline{\xi}^{-1}  Q
$. Let introduce the group
$$
{\rm Map}(X,\n{U}(2))_{\Z_2}\;:=\;\big\{\xi\in {\rm Map}(X,\n{U}(2))\ |\  \tau^*\xi=-Q  \overline{\xi}^{-1}  Q\big\}\;\subset {\rm Map}(X,\n{U}(2))\;.
$$
Two Q-structures $\Theta$ and $\Theta'$ on 
$X \times \C^2$, induced respectively by the maps $\xi$ and $\xi'$ in ${\rm Map}(X,\n{U}(2))_{\Z_2}$, are isomorphic if there exists a map $\psi\in{\rm Map}(X,\n{U}(2))$ such that
$\tau^*\psi \xi'  Q=\xi  Q\overline{\psi}$. Consider the action of 
${\rm Map}(X,\n{U}(2))$ on ${\rm Map}(X,\n{U}(2))_{\Z_2}$
defined as follows: For any $\psi\in {\rm Map}(X,\n{U}(2))$ let $G_\psi$ be the automorphism of ${\rm Map}(X,\n{U}(2))_{\Z_2}$ given by
the formula \eqref{eq:def_G}.
From the argument above it follows that
$$
{\rm Vec}^{2}_{Q}(X, \tau)\;\simeq\; {\rm Map}(X,\n{U}(2))_{\Z_2}\;/\;{\rm Map}(X,\n{U}(2))
$$
where the equivalence relation is induced by the action of the automorphisms $G_\psi$.
Since $H^2_{\Z_2}(X,\Z(1))=0$ by hypothesis any \virg{Real} line-bundle over $X$ is
 automatically trivial \cite{kahn-59}.
 This applies  in particular  to determinant line-bundle of the Q-bundle $(\bb{E},\Theta)$.
The triviality of the 
\virg{Real} structure $(x, u) \mapsto (\tau(x), {\rm det}(\xi)(x)\;\overline{u})$ on $X\times\C$ implies the existence of a map $\phi:X\to\n{U}(1)$ such that ${\rm det}(\xi)=\tau^*\phi\;\phi$. Consider the map
$\psi_0\in {\rm Map}(X,\n{U}(2))$ given by
$$
\psi_0(x)\;:=\;
\left(
\begin{array}{cc}
\phi(x) & 0 \\
0 & 1
\end{array}
\right).
$$
A direct computation shows that
\begin{equation}\label{eq:axox1}
{\rm det}(G_{\psi_0}(\xi))\;=\;{\rm det}(\tau^*\psi_0)^{-1}\; {\rm det}(\xi)\;{\rm det}(\psi_0)^{-1}\;=\;1.
\end{equation}
As a result, it is possible to choose $\xi \in {\rm Map}(X,\n{U}(2))_{\Z_2} \cap{\rm Map}(X,\n{SU}(2))$ as  representatives for the element of  ${\rm Vec}^{2}_{Q}(X, \tau)$.
Since it holds that 
$-Q  \overline{\xi}   Q={\xi}$ for maps with values in $\n{SU}(2)$, one has that the intersection ${\rm Map}(X,\n{U}(2))_{\Z_2} \cap{\rm Map}(X,\n{SU}(2))$
coincides with the group ${\rm Map}(X,\n{SU}(2))_{\Z_2}$ as described by \eqref{eq:coxxi-01}.
Finally, it is straightforward to see that the group ${\rm Map}(X, \n{U}(2))'_{\Z_2}$ 
described by \eqref{eq:coxxi-02}
is the maximal subgroup of ${\rm Map}(X, \n{U}(2))_{\Z_2}$ preserving such representatives.
\end{proof}

\medskip
\noindent
As a byproduct of the bijection \eqref{eq:first_biject} one may 
 think of ${\rm Vec}^{2}_{Q}(X, \tau)$ as a group with group structure  inherited from ${\rm Map}(X, \n{SU}(2))_{\Z_2}$.

\begin{lemma}
\label{lemma:prep2}
Under the hypotheses of
Lemma \ref{lemma:prep1}
 there is a natural group isomorphism
\begin{equation}\label{eq:first_grup_iso}
{\rm Vec}^{2}_{Q}(X, \tau)\;\simeq\; [X,\n{SU}(2)]_{\Z_2}\;/\;[X,\n{U}(2)]'_{\Z_2}\;.
\end{equation}
\end{lemma}
\begin{proof}
Consider the natural surjection onto the equivalence classes
$$
\varpi\; :\; {\rm Map}(X, \n{SU}(2))_{\Z_2}\; \hooklongrightarrow\; [X, \n{SU}(2)]_{\Z_2}\;.
$$
The action of ${\rm Map}(X, \n{U}(2))'$ on ${\rm Map}(X, \n{SU}(2))_{\Z_2}$ given by 
\eqref{eq:def_G}
induces an action of the group $[X, \n{U}(2)]'_{\Z_2}$ 
 on $[X, \n{SU}(2)]_{\Z_2}$. Under these actions, $\varpi$ is equivariant, and one gets
$$
{\rm Vec}^{2}_{Q}(X, \tau)\;\simeq\; {\rm Map}(X,\n{SU}(2))_{\Z_2}\;/\;{\rm Map}(X,\n{U}(2))'_{\Z_2}\;
\overset{\varpi}{\longrightarrow}\;
[X, \n{SU}(2)]_{\Z_2}\;/\;[X, \n{U}(2)]'_{\Z_2}\;.
$$
The latter is  an isomorphism of groups: Given $\xi \in {\rm Map}(X, \n{SU}(2))_{\Z_2}$, let $\bb{E}_\xi = X \times \C^2$ be the  Q-bundle of rank 2 with Q-structure given by $(x, v) \mapsto (\tau(x), \xi(x)Q\overline{v})$. 
In view of the  homotopy property of Q-bundles
if $\xi$ and $\xi'$ are $\Z_2$-homotopy equivalent, then $\bb{E}_\xi$ and $\bb{E}_{\xi'}$ are isomorphic.
Therefore one gets the map
$$
[X, \n{SU}(2)]_{\Z_2}\;/\;[X, \n{U}(2)]'_{\Z_2} \;\longrightarrow\; {\rm Vec}^{2}_{Q}(X, \tau)
$$
which is the  inverse to $\varpi$. 
\end{proof}

\medskip

We are now in position to complete the proof of 
Theorem \ref{theo:main-SU2_charat}. To that end
 the restriction to dimensions $d\leqslant 2$
will be crucial.

\proof[{Proof of Theorem \ref{theo:main-SU2_charat}}]
Consider the  exact sequence
$$
1 \;\longrightarrow\; {\rm Map}(X, \n{SU}(2)) \;\stackrel{\imath}{\longrightarrow}\; {\rm Map}(X, \n{U}(2))'_{\Z_2}\;\stackrel{{\rm det}}{\longrightarrow}\; {\rm Map}(X, \n{U}(1))_{\Z_2} \;\longrightarrow\; 1.
$$
where the map $\imath$ is the natural inclusion. Consider also the  
grup homomorphism $s : {\rm Map}(X, \n{U}(1))_{\Z_2} \to {\rm Map}(X, \n{U}(2))'_{\Z_2}$ given by
\begin{equation}\label{eq:homom_s}
{\rm Map}(X, \n{U}(1))_{\Z_2}\;\ni\;\phi\; \stackrel{s}{\longmapsto}
\left(
\begin{array}{cc}
\phi & 0 \\
0 & 1
\end{array}
\right)\;\in\; {\rm Map}(X, \n{U}(2))'_{\Z_2}\;.
\end{equation}
Since $\det \circ s = {\rm Id}$ the exact sequence is right-split and one has the group isomorphism
$$
{\rm Map}(X, \n{U}(2))'_{\Z_2}\;\simeq\;{\rm Map}(X, \n{SU}(2)) \;\rtimes\;{\rm Map}(X, \n{U}(1))_{\Z_2}
$$
where $\rtimes$ denotes the semi-drect product.
The whole  construction passes through the equivalence relation induced  by the  $\Z_2$-homotopy. Thus, one has the right-split exact sequence 
$$
1 \;\longrightarrow\; [X, \n{SU}(2)] \;\stackrel{\imath}{\longrightarrow}\; [X, \n{U}(2)]'_{\Z_2}\;\stackrel{{\rm det}}{\longrightarrow}\; [X, \n{U}(1)]_{\Z_2} \;\longrightarrow\; 1.
$$
and the group isomorphism
$$
[X, \n{U}(2)]'_{\Z_2}\;\simeq\;[X, \n{SU}(2)] \;\rtimes\;[X, \n{U}(1)]_{\Z_2}\;.
$$
Since $\pi_k(\n{SU}(2))=0$ if $k=0,1,2$ it follows that $[X, \n{SU}(2)]=0$ whenever $X$ has dimension $0\leqslant d \leqslant 2$. In the latter case the isomorphism above reduces to 
$[X, \n{U}(2)]'_{\Z_2}\;\simeq\;[X, \n{U}(1)]_{\Z_2}$ and the combination of the action $G$ described by \eqref{eq:def_G} with the homomorphism $s$ in \eqref{eq:homom_s} produces the action $L$ of $[X, \n{U}(1)]_{\Z_2}$ on 
$[X, \n{SU}(2)]_{\Z_2}$ as described in the claim.\qed

\begin{remark}[Higher rank case]\label{rk:higer-rk}{\upshape
In view of the stable rank condition described by Theorem \ref{theo:stab_ran_Q_even} the isomorphism proved in Theorem \ref{theo:main-SU2_charat} generalizes to 
$$
{\rm Vec}^{2m}_{Q}(X, \tau)\;\simeq\; [X,\n{SU}(2)]_{\Z_2}\;/\;[X,\n{U}(1)]_{\Z_2}\;,\qquad\quad m\in\N\;.
$$
A representative map $\xi:X\to \n{SU}(2)$  for a given  Q-bundle $(\bb{E},\Theta)$ of rank $2m$ can be constructed in this way: The Q-structure of 
$(\bb{E},\Theta)$ is coded in an equivariant map $\xi':X\to \n{SU}(2m)$ which, for instance, can be constructed from a global frame according to the prescription described in Remark \ref{rk:map-frame}. The stable rank condition  implies that 
$\xi'$ can be always reduced in the form
$$
\xi'\;\simeq\left(\begin{array}{c|c}\xi & 0 
\\
\hline0 & \n{1}_{\C^{2(m-1)}}\end{array}\right)
$$
up to the conjugation with an equivariant map with values in $\n{U}(2m)$. The reduced map 
$\xi:X\to \n{SU}(2)$ obtained in this way provides the representative of the Q-bundle $(\bb{E},\Theta)$ as element of the group $[X,\n{SU}(2)]_{\Z_2}/[X,\n{U}(1)]_{\Z_2}$.
}\hfill $\blacktriangleleft$
\end{remark}

  %----------------------------%
\subsection{The FKMM-invariant for oriented two-dimensional FKMM-manifold}
\label{subsec:altern_FKMM-2D}
Throughout this section we will assume that the pair $(\Sigma,\tau)$ is an oriented two-dimensional FKMM-manifold in the sense of Definition \ref{def:good_manif_d=2}. The use of the letter $\Sigma$ instead of $X$ is motivated 
to easier connect the results discussed here with the theory developed in Section \ref{sec:Wess_Zumono_term}, Section \ref{sec:Wess_Zumono_term_no-bound} and Section \ref{sec:Wess_Zumono_term_class_D2} 

\medskip

When $(\Sigma,\tau)$ is an oriented  two-dimensional FKMM-manifold then two presentations for 
${\rm Vec}^{2}_{Q}(\Sigma, \tau)$ are available.
The first one
$$
{\rm Vec}^{2}_{Q}(\Sigma, \tau)\;\simeq\; {\rm Map}(\Sigma^\tau,\{\pm1\})\;/\;[\Sigma,\n{U}(1)]_{\Z_2}
$$
has been proved in  Proposition  \ref{prop:fkmm-inv_fkmm-space} and
uses the FKMM-invariant. The second one
$$
{\rm Vec}^{2}_{Q}(\Sigma, \tau)\;\simeq\; [\Sigma,\n{SU}(2)]_{\Z_2}\;/\;[\Sigma,\n{U}(1)]_{\Z_2}
$$
comes from Theorem \ref{theo:main-SU2_charat}.
Therefore, it must exist an isomorphism of groups
$$
[\Sigma,\n{SU}(2)]_{\Z_2}\;/\;[\Sigma,\n{U}(1)]_{\Z_2}\;\simeq\;{\rm Map}(\Sigma^\tau,\{\pm1\})\;/\;[\Sigma,\n{U}(1)]_{\Z_2}
$$
which associates the map $\xi\in {\rm Map}(\Sigma,\n{SU}(2))_{\Z_2}$ with the FKMM-invariant 
of the 
Q-bundle $\bb{E}_\xi$ classified by $\xi$. Such a map can be constructed by means of the \emph{Pfaffian} ${\rm Pf}$ (\cf Proposition \ref{prop:classification}).

\medskip

Every map $\xi\in {\rm Map}(\Sigma,\n{SU}(2))_{\Z_2}$ when evaluated on a fixed point $x\in\Sigma^\tau$ gives rise to a $\n{SU}(2)$ matrix which verifies $\xi(x)=\xi(x)^{-1}$. This implies that  $\xi(x)=\pm\n{1}_{\C^2}$ if $x\in\Sigma^\tau$. Moreover,
every matrix $\xi(x)\in \n{SU}(2)$ verifies the identity $Q \overline{\xi(x)}={\xi}(x)  Q$. Then, on a  fixed point $x\in\Sigma^\tau$ the matrix ${\xi}(x) Q=\pm Q$ turns out to be skew-symmetric and the {Pfaffian} ${\rm Pf}({\xi}(x)  Q)$ results well-defined. In particular one has that
$$
-{\rm Pf}({\xi}(x) \cdot Q)\;=\;
\left\{
\begin{aligned}
&+1&\qquad&\text{if}\quad {\xi}(x)=+\n{1}_{\C^2}\\
&-1&\qquad&\text{if}\quad {\xi}(x)=-\n{1}_{\C^2}\;.
\end{aligned}
\right.
$$
This suggests  to study
 the following mapping 
\begin{equation}
{\rm Map}(\Sigma,\n{SU}(2))_{\Z_2}\;\ni\;\xi\;\stackrel{\Phi_\kappa}{\longrightarrow}\;-{\rm Pf}({\xi}   Q)|_{\Sigma^\tau}\;\in\;{\rm Map}(\Sigma^\tau,\{\pm1\})
\end{equation}
\begin{lemma}
\label{lemm:iso_for_class}
Let $(\Sigma,\tau)$ be an oriented two-dimensional FKMM-manifold in the sense of Definition \ref{def:good_manif_d=2}. Then, there is  an isomorphism of groups
$$
\Phi_\kappa\;:\;[\Sigma,\n{SU}(2)]_{\Z_2}\;\longrightarrow\;{\rm Map}(\Sigma^\tau,\{\pm1\})
$$
defined by $[\xi]\mapsto -{\rm Pf}({\xi}Q)|_{\Sigma^\tau}$.
\end{lemma}
\proof
Since $\Phi_\kappa$ is by construction a homomorphism of groups it is necessary to prove only the injectivity and the surjectivity of $\Phi_\kappa$.
Let us start with the injectivity.
For that it suffices to show that every 
$\xi\in {\rm Map}(\Sigma,\n{SU}(2))_{\Z_2}$ such that $\xi(\Sigma^\tau)=\{\n{1}_{\C^2}\}$ 
is $\Z_2$-homotopy equivalent to the constant map at $\n{1}_{\C^2}$. This is a problem in equivariant homotopy theory, and the criterion for \emph{$\Z_2$-homotopy reduction} 
proved in \cite[Lemma 4.27]{denittis-gomi-14} can be used. 
Under the involution $\imath : \n{SU}(2) \to \n{SU}(2)$ defined by the inverse  $\imath(g):=g^{-1}$ the fixed point set $\n{SU}(2)^\imath$ consists of $\pm \n{1}_{\C^2}$. Hence $\pi_0(\n{SU}(2)^\imath) \simeq \Z_2$ and $\pi_k(\n{SU}(2)^\imath) = 0$ for all $k > 0$. Also, $\pi_k(\n{SU}(2)) = 0$ for $0 \leqslant k  \leqslant 2$. By assumption, $\Sigma$ admits the structure of a $\Z_2$-CW complex with possible free cells  only in dimension $0$.
 Though $\pi_0(\n{SU}(2)^\imath) \neq 0$, one can use the initial constraint $\xi(\Sigma^\tau)=\{\n{1}_{\C^2}\}$ to start the inductive argument in  \cite[Lemma 4.27]{denittis-gomi-14}. The result is that $\xi$ can be equivariantly deformed to the constant map at $\n{1}_{\C^2}$, proving in this way the injectivity of $\Phi_\kappa$.\\
 Now the surjectivity.
The idea is to construct an element  $\xi_\epsilon \in {\rm  Map}(\Sigma, \n{SU}(2))_{\Z_2}$ for each $\epsilon \in {\rm Map}(\Sigma^\tau, \Z_2)$ such that $\Phi_\kappa(\xi_\epsilon )=\epsilon$.
A preliminary fact is necessary.
Let $D\subset\C$ be the closed unit disk endowed  with the involution $z \mapsto -z$. Then, the map $\xi_{D} \in{\rm Map}(D, \n{SU}(2))_{\Z_2}$ given by
$$
\xi_{D}(z)\;:=\;\frac{1}{2(|z|^2-|z|)+1}\left(
\begin{array}{cc}
2|z|-1 & -2\overline{z}(|z|-1) \\
2z(|z|-1) & 2|z|-1
\end{array}
\right)
$$
verifies $\xi_{D}(0)=-\n{1}_{\C^2}$ and
 $\xi_{D}(z)=+\n{1}_{\C^2}$ if $z\in\partial D$.
Let $\Sigma^\tau = \{ x_1, \ldots, x_n \}$ be a given
labeling 
 for the fixed points. The \emph{slice theorem} \cite[Chapter I, Section 3]{hsiang-75} 
assures that for each $x_i$ 
there exists a closed disk $D_i \subset \Sigma$ such that $\tau(D_i) = D_i$, $x_i \in D_i$ and $D_i \cap D_j = \emptyset$ when $i \neq j$. 
Let $x_{i_1}, \ldots, x_{i_k} \in \Sigma^\tau$ be the set of  points such that $\epsilon(x_{i_1}) = -1$. Using an equivariant diffeomorphism $D \cong D_{i_j}$ one can induce the equivariant map
$\xi_{D_{i_j}}$ on $D_{i_j}$  from $\xi_D$. Extending these maps by $\n{1}_{\C^2}$  outside of $D_{i_1} \cup \cdots \cup D_{i_k}$ one gets an equivariant map $\xi_\epsilon \in{\rm Map}(\Sigma, \n{SU}(2))_{\Z_2}$ such that $\xi_\epsilon(x) = \epsilon(x) \n{1}_{\C^2}$ for every $x \in \Sigma^\tau$. This ensures that $\Phi_\kappa(\xi_\epsilon )=\epsilon$.
\qed

\begin{proposition} 
\label{prop:classification}
Let $(\Sigma,\tau)$ be an oriented two-dimensional FKMM-manifold in the sense of Definition \ref{def:good_manif_d=2}. Then, the isomorphism of Lemma \ref{lemm:iso_for_class} induces the 
 isomorphism of groups
$$
\Phi_\kappa\;:\; [\Sigma,\n{SU}(2)]_{\Z_2}\;/\;[\Sigma,\n{U}(1)]_{\Z_2}\;\longrightarrow\;{\rm Map}(\Sigma^\tau,\{\pm1\})\;/\;[\Sigma,\n{U}(1)]_{\Z_2}\;.
$$
\end{proposition}
\begin{proof}
Lemma \ref{lemm:iso_for_class}  asserts the bijectivity of the homomorphism
$$
\Phi_\kappa\;:\;[\Sigma,\n{SU}(2)]_{\Z_2}\;\longrightarrow\;{\rm Map}(\Sigma^\tau,\{\pm1\})
\;.
$$
On both sides the same group $[\Sigma, \n{U}(1)]_{\Z_2}$ acts and  $\Phi_\kappa$ turns out to be equivariant.
An inspection of the group  actions shows that 
 $\Phi_\kappa$ descends to a bijective homomorphism  between the quotients.
\end{proof}

\medskip

In view of Theorem \ref{theo:main-SU2_charat}, one can think of a map $\xi \in{\rm Map}(\Sigma, \n{SU}(2))_{\Z_2}$ as a rank $2$ Q-bundle on $\Sigma$. Then, it makes sense to talk about the \emph{\virg{FKMM-invariant of the map $\xi$}}.
Proposition \ref{prop:classification} shows that such an invariant is indeed built through the isomorphism  $\Phi_\kappa$. More precisely, by combining Proposition \ref{prop:classification} with  Theorem \ref{prop:Fu-Kane-Mele_formula1} one obtains that
\begin{equation}
\label{eq:FKMM_from_map}
\kappa(\xi)\;:=\;\Pi\circ\Phi_\kappa(\xi)\;\in\;\Z_2 
\end{equation}
where $\kappa(\xi)$ has the meaning of {the FKMM-invariant of the Q-bundle defined by the map $\xi$}.

\begin{remark}[Construction of the classifying map from a frame]\label{rk:map-frame}{\upshape
Let $(\bb{E},\Theta)$ be a  Q-bundle of rank 2 over an oriented two-dimensional FKMM-manifold. If the map $\xi \in{\rm Map}(\Sigma, \n{SU}(2))_{\Z_2}$  classifies  $(\bb{E},\Theta)$ according to Theorem \ref{prop:Fu-Kane-Mele_formula1} then formula \eqref{eq:FKMM_from_map} provides the computation of the FKMM-invariant of $(\bb{E},\Theta)$. Therefore, the relevant problem is how to extract $\xi$ from the  knowledge of 
$(\bb{E},\Theta)$. This problem has a simple solution when a global trivializing frame of sections $t_1,t_2:\Sigma\to \bb{E}$ of the underlying (trivial) complex vector bundle is known. This situation has been described in detail   \cite[Section 4.2]{denittis-gomi-14-gen}. By a Gram-Schmidt orthonormalization if necessary, one can assume without loss of generality that the frame $t_1,t_2$ is orthonormal, \ie $\rr{m}(t_i,t_j)=\delta_{i,j}$
where $\rr{m}$ is the (unique) {$\Theta$- equivariant} Hermitian metric on $\bb{E}$.
Then the classifying map $\xi=\{\xi_{ij}\}$ is given by the formula
$$
\xi_{ij}(x)\;:=\;\rr{m}\big(\tau^*t_i(x),\Theta t_j(x)\big)
$$
where $\tau^*t_i(x):=t_i(\tau(x))$ and $\Theta t_j(x):=\Theta(t_j(x))$ are short notations.
 }\hfill $\blacktriangleleft$
\end{remark}

%--------------------%
%--------------------%
%--------------------%
\section{Differential geometric classification of \virg{Quaternionic} vector bundles}
\label{sect:geom_class}
In this section we provide  differential geometric realizations of the FKMM-invariant. However, this require some more structure on the involutive space $(X,\tau)$. More properly we need to pass from the \emph{topological category}  to the \emph{smooth category} . In this section  the quite general Assumption \ref{ass:top} will be replaced  by the more restrictive
\begin{assumption}[Smooth category]\label{ass:smooth}
$X$ is a  compact,  path-connected, Hausdorff \emph{smooth}  $d$-dimensional manifold without boundary and with a  \emph{smooth} involution $\tau$
\end{assumption}
\noindent
In particular, a space  $X$ which fulfills Assumption \ref{ass:smooth}
 is a \emph{closed} manifold and the pair $(X,\tau)$ automatically admits the structure of a $\Z_2$-CW-complex (see \eg \cite[Theorem 3.6]{may-96}). Observe that the notion of FKMM-manifold given in Definition \ref{def_FKMM-space} is compatible with Assumption \ref{ass:smooth}. It is worth point out  that the \emph{smooth} condition  can be relaxed to a less demanding regularity condition; For instance is sufficient to assume that the manifold structure is $C^r$-regular for some $r\in\N$. Anyway,  this is only a technical detail and for a simpler presentation  it is enough to focus only on the smooth case.

\medskip

 Let us point out that in Section \ref{sect:basic_def}  we introduced the 
 notion of Q-bundle 
  in the {topological category} meaning that  all the maps involved in the various definitions are continuous functions between topological spaces. However, when the involutive space $(X,\tau)$ has an additional smooth 
manifold structure one can equivalently define Q-bundles in the {smooth category}
by requiring that all  spaces involved in the definitions carry a smooth
manifold structure and maps are smooth functions. However, for what concerns the problem of the classification the two categories are equivalent \cite[Theorem 2.1]{denittis-gomi-15}, namely
$$
^{\rm top}{\rm Vec}^m_{Q}(X,\tau)\;\simeq\;^{\rm smooth}{\rm Vec}^m_{Q}(X,\tau)\;.
$$
Clearly, the same holds true also in the \virg{Real} category.
For more details on this point we refer to \cite[Section 2]{denittis-gomi-15}.

%------%
\subsection{\virg{Quaternionic} principal bundles and related FKMM-invariant
}
\label{sect:Q_princ_bund&FKMM}
The definition of \virg{Quaternionic} principal bundle (or principal Q-bundle) has been introduced in \cite[Section 2.1]{denittis-gomi-15}. Before giving the formal definition let us recall that principal bundles are related to vector bundles through the \emph{structure group}.
Since  \virg{Quaternionic} (as well as \virg{Real}) vector bundles over a compact base space admit an equivariant Hermitian metric (see \cite[Remark 4.11]{denittis-gomi-14} and \cite[Proposition 2.10]{denittis-gomi-14-gen}), it turns out that the relevant structure group is the {unitary group} $\n{U}(m)$ together with its Lie algebra $\rr{u}(m)$ consisting of anti-Hermitian matrices. 
For a concise summary about the theory of principal bundles we refer to \cite[Appendix B]{denittis-gomi-15} and references therein.

\begin{definition}[Principal R- and Q-bundle]\label{def_princ_R&QB}
Let $(X,\tau)$ be an involutive space which verifies Assumption \ref{ass:smooth}  and $\pi:\bb{P}\to X$ a 
(smooth) principal $\n{U}(m)$-bundle.
 We say that $\bb{P}$ has a \emph{\virg{Real} structure} if there is a  homeomorphism $\hat{\Theta}:\bb{P}\to \bb{P}$ such that:
\begin{enumerate}
\item[\upshape{(Eq.)}] The bundle projection $\pi$ is \emph{equivariant} in the sense that $\pi\circ\hat{\Theta}=\tau\circ\pi$;\vspace{1.2mm}
\item[\upshape{(Inv.)}] $\hat{\Theta}$ is an \emph{involution}, \ie $\hat{\Theta}^2(p)=p$ for all $p\in\bb{P}$;
\vspace{1.2mm}
\item[{\upshape ($\hat{R}$)}] The right $\n{U}(m)$-action on the fibers and the homeomorphism $\hat{\Theta}$ fulfill the condition
$$
\hat{\Theta}\big(R_u(p)\big)\;=\;R_{\overline{u}}\big(\hat{\Theta}(p)\big)\;,\qquad\quad \forall\ p\in\bb{P}\;,\ \ \ \forall\ u\in \n{U}(m)
$$
where $R_u(p)=p\cdot u$ denotes the right $\n{U}(m)$-action and $\overline{u}$ is the complex conjugate of $u$.
\end{enumerate}
We say that $\bb{P}$ has a \emph{\virg{Quaternionic} structure} if 
the structure group $\n{U}(2m)$ has even rank and condition
{\upshape ($\hat{R}$)} is replaced by
\begin{enumerate}
\item[{\upshape ($\hat{Q}$)}] The right $\n{U}(2m)$-action on the fibers and the homeomorphism $\hat{\Theta}$ fulfill the condition
$$
\hat{\Theta}(R_u(p))\;=\;R_{\sigma(u)}(\hat{\Theta}(p))\;,\qquad\quad \forall\ p\in\bb{P}\;,\ \ \ \forall\ u\in \n{U}(2m)
$$
where $\sigma:\n{U}(2m)\to\n{U}(2m)$ is the involution given by 
$$
\sigma(u)\;:=\;Q\cdot\overline{u}\cdot Q^{-1}\;=\;-Q\cdot\overline{u}\cdot Q
$$ 
and $Q$ is the matrix \eqref{eq:Q-mat}.
\end{enumerate}
\end{definition} 

\noindent
We will often refer to \virg{Real} and \virg{Quaternionic} principal bundles with the abbreviations principal R-bundles
and principal Q-bundles, respectively. 
\begin{remark}\label{rk:quat_str}{\upshape
Let us notice  that both   the \virg{Real} and  the \virg{Quaternionic} case require that  $\hat{\Theta}$ has to be an involution as imposed by the  property (Inv.). This means that both principal  R- and Q-bundles are examples of \emph{$\Z_2$-equivariant} 
principal bundles (indeed
properties (Eq.) and (Inv.) define these objects). This is indeed a difference with respect to the vector bundle case  (\cf with Definition \ref{defi:Q_VB}).
}\hfill $\blacktriangleleft$
\end{remark}

Morphisms (and isomorphisms) between  principal  R- and Q-bundles are defined in a natural way: If $(\bb{P},\hat{\Theta})$ and $(\bb{P}',\hat{\Theta}')$ are two of such principal bundles over the same involutive space $(X,\tau)$ then an R- or Q-morphism  is a principal bundle morphism $f:\bb{P}\to\bb{P}'$
 such that $f\circ\hat{\Theta}'=\hat{\Theta}\circ f$. We will use the symbols  
${\rm Prin}_{R}^{\n{U}(m)}(X,\tau)$ and ${\rm Prin}_{Q}^{\n{U}(2m)}(X,\tau)$ for the sets of equivalence classes of \virg{Real} and \virg{Quaternonic} principal bundles over  $(X,\tau)$, respectively. An principal R-bundle over $(X,\tau)$
is called \emph{trivial} if it is isomorphic to the {product}  bundle $X\times\n{U}(m)$ with \emph{trivial} R-structure $\hat{\Theta}_0:(x,u)\mapsto(\tau(x),\overline{u})$. In much the same way, a \emph{trivial} principal Q-bundle is isomorphic to the
{product}  bundle  $X\times\n{U}(2m)$ endowed with the  \emph{trivial} Q-structure $\hat{\Theta}_0:(x,u)\mapsto(\tau(x),\sigma(u))$. 

\medskip

A standard result says that there is an equivalence of categories between  principal $\n{U}(m)$-bundles  and complex vector bundles. This equivalence is realized by the \emph{associated bundle} construction along its inverse, called \emph{orthonormal frame bundle} construction (see \cite[Appendix B]{denittis-gomi-15} for more details). A similar result extends to the \virg{Real} and the\virg{Quaternonic} categories \cite[Proposition 2.4]{denittis-gomi-15} leading to 
\begin{equation}\label{eq:iso_vect_princ_QR}
{\rm Prin}_{R}^{\n{U}(m)}(X,\tau)\;\simeq\; {\rm Vec}^m_{R}(X,\tau)\;,\qquad\quad {\rm Prin}_{Q}^{\n{U}(2m)}(X,\tau)\;\simeq\; {\rm Vec}^{2m}_{Q}(X,\tau)\;.
\end{equation}
We can take advantage of the above isomorphisms to carry the notion of FKMM-invariant from
vector bundles to principal bundles.
\begin{definition}[FKMM-invariant: principal bundle version]\label{def:gen_FKMM_inv_princ_bund}
Let $(\bb{P},\hat{\Theta})$ be a 
rank $2m$ principal
Q-bundle
over the involutive space $(X,\tau)$. Let $[(\bb{E},\Theta)]\in {\rm Vec}^{2m}_{Q}(X,\tau)$ be the unique class associated with $[(\bb{P},\hat{\Theta})]\in {\rm Prin}_{Q}^{\n{U}(2m)}(X,\tau)$ by the isomorphism \eqref{eq:iso_vect_princ_QR}.
One defines the  \emph{FKMM-invariant} of  $(\bb{P},\hat{\Theta})$ as the  
FKMM-invariant of the associate Q-bundle  $(\bb{E},\Theta)$, namely
$$
\kappa(\bb{P},\hat{\Theta})\;:=\;\kappa(\bb{E},\Theta)\;.
$$
\end{definition}
\begin{remark}\label{rk:quat_str2}{\upshape
Let us briefly discuss the consistency of Definition \ref{def:gen_FKMM_inv_princ_bund} with
the construction of the FKMM-invariant presented in \cite{denittis-gomi-16}.
In view of the isomorphisms \eqref{eq:iso_vect_princ_QR}  to each $\n{U}(2m)$ principal Q-bundle
$(\bb{P},\hat{\Theta})$ one can associate a unique (up to isomorphisms)  $\n{U}(1)$ principal R-bundle $({\rm det}(\bb{P}),{\rm det}(\hat{\Theta}))$ which is defined as the unique (up to isomorphisms) $\n{U}(1)$ principal R-bundle associated with the rank one R-bundle $({\rm det}(\bb{E}),{\rm det}(\Theta))$. Moreover, there is a one-to-one correspondence
 between sections of a  $\n{U}(1)$ principal R-bundle and sections of a rank one R-bundle.
Then, the quantity $\kappa(\bb{P},\hat{\Theta})$ turns out to be determined by the equivalence class of the
pair $({\rm det}(\bb{P}),s_{\bb{P}})$ where $s_{(\bb{P},\hat{\Theta})}\equiv s_{(\bb{E},\Theta)}$ is the canonical section associated to $(\bb{E},\Theta)$. For more details about the relation between the FKMM-invariant and the canonical section we refer to \cite[Section 3.2]{denittis-gomi-14-gen} or \cite[Section 2.2]{denittis-gomi-16}.
}\hfill $\blacktriangleleft$
\end{remark}
%

%-----------------%
\subsection{\virg{Quaternionic} connections and curvatures}
\label{sect:Q-&connect}
Connections with \virg{Quaternionic}  and \virg{Real} structures have been studied in \cite[Section 2.2]{denittis-gomi-15}. We review here the basic definitions and the main properties of these objects.
For a reminder about the theory of connections we refer to the classic monographs \cite{kobayashi-nomizu,kobayashi-87} (see also \cite[Appendix B]{denittis-gomi-15} and references therein).

\medskip

 We consider  principal bundles in the smooth category $\pi:\bb{P}\to X$ endowed with a \virg{Real} or \virg{Quaternionic} structure $\hat{\Theta}:\bb{P}\to \bb{P}$ over
 the involutive space $(X,\tau)$. The structure group is $\n{U}(m)$ ($m$ even in the \virg{Quaternionic} case) and $\rr{u}(m)$ is the related Lie algebra.
The symbol $\omega\in\Omega^1(\bb{P},\rr{u}(m))$ will be used for  the \emph{connection} 1-forms associated to  given  horizontal distributions $p\mapsto H_p$ of $\bb{P}$ . We observe that the Lie algebra $\rr{u}(m)$ has two natural involutions: a \emph{real} involution $\rr{u}(m)\ni\xi\mapsto \overline{\xi}\in\rr{u}(m)$ and a \emph{quaternionic} involution $\rr{u}(2m)\ni \xi\mapsto \sigma(\xi):=-Q\cdot\overline{\xi}\cdot Q\in \rr{u}(2m)$. Here $\xi\in \rr{u}(m)$ is any anti-Hermitian matrix of size $m$ and the matrix $Q$ has been defined in \eqref{eq:Q-mat}. Finally, given a $k$-form $\phi \in\Omega^k(\bb{P}, \bb{A})$ with value in some structure $\bb{A}$ (module, ring, algebra, group, \etc) and a smooth map $f: \bb{P}\to \bb{P}$ we denote with $f^*\phi:=\phi\circ f_\ast$ the \emph{pull-back} of  $\phi$ with respect to the map $f$ (and $f_\ast:T\bb{P}\to T\bb{P}$ is the differential, or \emph{push-forward}, of vector fields).
Given a $\rr{u}(m)$-valued $k$-form $\phi \in\Omega^k(\bb{P}, \rr{u}(m))$ we define the \emph{complex conjugate} form $\overline{\phi}$ pointwise, \ie 
$\overline{\phi}_p({\rm w}^1_p,\ldots,{\rm w}^k_p):=\overline{{\phi}_p({\rm w}^1_p,\ldots,{\rm w}^k_p)}$ for every $k$-tupla $\{{\rm w}^1_p,\ldots,{\rm w}^k_p\}\subset T_p\bb{P}$ of tangent vectors at $p\in\bb{P}$. It follows that $f^*\overline{\phi}=\overline{f^*\phi}$ for every smooth map $f:\bb{P}\to\bb{P}$. Similarly, if $\phi \in\Omega^k(\bb{P}, \rr{u}(2m))$ we define $\sigma(\phi)$ pointwise by
$\sigma(\phi)_p({\rm w}^1_p,\ldots,{\rm w}^k_p):=-Q\cdot\overline{{\phi}_p({\rm w}^1_p,\ldots,{\rm w}^k_p)}\cdot Q$. Hence, one has that $\sigma(f^*\phi)=f^*\sigma(\phi)$.
With these premises we are now in position to give the following definitions.

\begin{definition}[\virg{Real} and \virg{Quaternionic} equivariant connections]\label{def:R&Q_connec}
Let $(X,\tau)$ be an involutive space that verifies Assumption \ref{ass:smooth} and $\pi:\bb{P}\to X$ a smooth principal $\n{U}(m)$-bundle over $X$
endowed with a \virg{Real} or a \virg{Quaternionic} structure $\hat{\Theta}:\bb{P}\to \bb{P}$ as in Definition   \ref{def_princ_R&QB}.
A {connection} 1-form $\omega\in\Omega^1(\bb{P},\rr{u}(m))$ is said to be \emph{equivariant} if $\overline{\omega}=\hat{\Theta}^\ast\omega$
in the \virg{Real} case or $\sigma(\omega)=\hat{\Theta}^\ast\omega$ in the \virg{Quaternionic} case. 
Equivariant connections in the \virg{Real} case are called \virg{Real} connections (or R-connections). Similarly, the \virg{Quaternionic} connections (or Q-connections) are the equivariant connections in the \virg{Quaternionic} category. 
\end{definition}

Let $\rr{A}_{R}(\bb{P})\subset \Omega^1(\bb{P},\rr{u}(m))$ be the space of R-connections on the  principal R-bundle $(\bb{P},\hat{\Theta})$. Similarly,  $\rr{A}_{Q}(\bb{P})\subset \Omega^1(\bb{P},\rr{u}(2m))$ will denote the space of Q-connections on the principal Q-bundle $(\bb{P},\hat{\Theta})$. Let us  introduce the sets of equivariant 1-forms
\begin{equation}\label{eq:set_equi_form}
\begin{aligned}
\Omega^1_{R}\big(\bb{P},\rr{u}(m)\big)\;&:=\left\{\omega\in \Omega^1(\bb{P},\rr{u}(m))\ |\ \overline{\omega}=\hat{\Theta}^*\omega\right\}\;\\
\Omega^1_{Q}\big(\bb{P},\rr{u}(2m)\big)\;&:=\left\{\omega\in \Omega^1(\bb{P},\rr{u}(2m))\ |\ \sigma(\omega)=\hat{\Theta}^*\omega\right\}\;.
\end{aligned}
\end{equation}
A 1-form is called \emph{horizontal} if it vanishes on vertical vectors. The set of $\rr{u}(m)$-valued 1-forms on $\bb{P}$ which are {horizontal} and which 
transform according to the adjoint representation of the structure group
 is denoted with $\Omega^1_{\rm hor}(\bb{P},\rr{u}(m),{\rm Ad})$. Let us introduce the sets
$$
\begin{aligned}
\s{V}^1_{R}\big(\bb{P}\big)\;&:=\; \Omega^1_{\rm hor}\big(\bb{P},\rr{u}(m),{\rm Ad}\big)\;\cap\; \Omega^1_{R}\big(\bb{P},\rr{u}(m)\big)\\
\s{V}^1_{Q}\big(\bb{P}\big)\;&:=\; \Omega^1_{\rm hor}\big(\bb{P},\rr{u}(2m),{\rm Ad}\big)\;\cap\; \Omega^1_{Q}\big(\bb{P},\rr{u}(2m)\big)\;.\\
\end{aligned}
$$
\begin{proposition}[{\cite[Proposition 2.11 \& Proposition 2.12]{denittis-gomi-15}}]
The sets $\rr{A}_{R}(\bb{P})$ and $\rr{A}_{Q}(\bb{P})$ are non-empty and are closed under convex combinations with real coefficients.
Moreover, they are affine spaces modeled on the vector spaces $\s{V}^1_{R}\big(\bb{P}\big)$ and $\s{V}^1_{Q}\big(\bb{P}\big)$, respectively.
\end{proposition}

\medskip

Connection 1-forms of a principal $\n{U}(m)$-bundles can be described in terms of collections of local 1-forms on the base space subjected to suitable gluing rules. This fact extends to the categories of \virg{Real} and \virg{Quaternionic} principal bundles, provided that an extra equivariance condition is added {\cite[Appendix B]{denittis-gomi-15}}.  Let $\pi:\bb{P}\to X$ be a  principal R or Q-bundle over the involutive space $(X,\tau)$ and consider an \emph{equivariant}  local trivialization $\{\f{U}_\alpha,h_\alpha\}$ (in the sense of 
\cite[Remark 2.6]{denittis-gomi-15}) with related transition functions $\{\varphi_{\beta,\alpha}\}$. On each open set  $\f{U}_\alpha\subset X$ we can define a local (smooth) section  $\rr{s}_\alpha(x):=h_\alpha^{-1}(x,\n{1})$ with $\n{1}_{\C^m}\in\n{U}(m)$ the identity matrix. 

\medskip

Let $F_\omega$ be the curvature associated
to the equivariant connections $\omega$
 by the 
\emph{structural equation} 
$$
F_\omega\;:=\;{\rm d}\omega\;+\;\frac{1}{2}\; [\omega\wedge \omega]\;.
$$
According to \cite[Proposition  2.22]{denittis-gomi-15}
one has that $F_\omega$ obeys to the  equivariant constraints:
\begin{equation}
\begin{aligned}
\overline{F_\omega}\;&=\;\hat{\Theta}^*{F_\omega}&&\qquad\text{(\virg{Real} case)}\\
\sigma(F_\omega)\;&=\;\hat{\Theta}^*F_\omega&&\qquad\text{(\virg{Quaternionic} case)}\;.\\
\end{aligned}
\end{equation}
Let $\{\s{F}_\alpha\in\Omega^2(\f{U}_\alpha,\rr{g})\}$ be the collection of local 2-forms which 
provides the local description of the 
 the curvature $F_\omega$ (in the sense of  \cite[Theorem C.2]{denittis-gomi-15}). When $\omega$ is equivariant it holds true that
\begin{equation}
\begin{aligned}
\overline{\s{F}_\alpha}\;&=\;{\tau}^*{\s{F}_\alpha}&&\qquad\text{(\virg{Real} case)}\\
\sigma(\s{F}_\alpha)\;&=\;{\tau}^*\s{F}_\alpha&&\qquad\text{(\virg{Quaternionic} case)}\;.\\
\end{aligned}
\end{equation}
%

%-----------------%
\subsection{Chern-Simons form and \virg{Quaternionic}structure}
\label{sect:chern-simons}
In this section we discuss some  aspect of 
 the Chern-Simons theory defined over (compact) manifolds without boundary in presence of a Q-structure.
For a comprehensive introduction to the Chern-Simons theory we refer to \cite{freed-95,hu-01}.

\medskip

Let  $\pi:\bb{P}\to X$ be a (smooth) principal $\n{U}(m)$-bundle and $\omega\in\Omega^1(\bb{P},\rr{u}(m))$ a {connection} 1-form.
 The \emph{Chern-Simons 3-form} $\f{CS}(\omega)\in\Omega^3(\bb{P})$ associated to $\omega$ is defined by 
 \begin{equation}\label{eq:Chern_Simons} 
 \f{CS}(\omega)\;:=\;\frac{1}{8\pi^2}\;{\rm Tr}\left(\omega\wedge\dd\omega\;+\;\frac{2}{3}\omega\wedge\omega\wedge\omega\right)
 \end{equation}
 where ${\rm Tr}$ is the usual trace on $m\times m$ matrices. The 3-form $\f{CS}(\omega)$ is sometimes called \emph{Chern-Simons Lagrangian}.
 A direct computation shows that the exterior differential $\dd\f{CS}(\omega)\in\Omega^4(\bb{P})$
can be expressed in terms of the curvature $F_\omega\in\Omega^2(\bb{P},\rr{u}(m))$ according to 
 \begin{equation}\label{eq:Chern_Simons_diff} 
 \dd\f{CS}(\omega)\;:=\;\frac{1}{4\pi^2}\;{\rm Tr}\big(F_\omega\wedge F_\omega\big)\;.
 \end{equation}
The following result will be used several times in the continuation of this work.
\begin{lemma}\label{lemma_CS_1}
Assume that $\pi:\bb{P}\to X$ admits a (smooth) section $s:X\to\bb{P}$ and let $g:X\to\n{U}(m)$ be a 
(smooth) map.
Define the new section $s_g:X\to\bb{P}$ through the right action $s_g(x):=R_{g(x)}(s(x))\equiv s(x)\cdot g(x)$. Then the two pull-backs  $s^*_g\f{CS}(\omega),s^*\f{CS}(\omega)\in\Omega^3(X)$ are related by the equation
\begin{equation}\label{eq:formul_CS_W}
s^*_g\f{CS}(\omega)\;=\; s^*\f{CS}(\omega)\;-\;\frac{1}{8\pi^2}\; \dd{\rm Tr}\big(s^*\omega\wedge\dd g^{-1} g\big)\;+\;\Lambda(g)
\end{equation}
where $\Lambda(g)\in\Omega^3(X)$ is given by
\begin{equation}\label{notat_Lmab}
\Lambda(g)\;:=\;-\frac{1}{24\pi^2}{\rm Tr}\left((g^{-1}\dd g)^{\wedge 3} \right)\;.
\end{equation}
\end{lemma}
\proof
The proof is essentially a computation which is based on the two relations: $s_g^*\f{CS}(\omega)=\f{CS}(s_g^*\omega)$ and $s_g^*\omega=g^{-1}(s^*\omega) g+g^{-1}\dd g$. Therefore, by exploiting the cyclicity of the trace, one can check that
$$
\f{CS}\left(g^{-1}(s^*\omega) g+g^{-1}\dd g\right)\;=\;\f{CS}\left(s^*\omega\right) \;-\;\frac{1}{8\pi^2}\; \dd{\rm Tr}\big(s^*\omega\wedge g^{-1} \dd g \big)\;-\; \frac{1}{24\pi^2}{\rm Tr}\left((g^{-1}\dd g)^{\wedge 3} \right)\;.
$$
The identity $0=\dd(g^{-1}g)=\dd g^{-1} g+g^{-1}\dd g$ concludes the computation.\qed 
 
\medskip

\begin{definition}[Chern-Simons invariant] 
\label{def:CS_inv}
Let $X$ be a compact oriented 3-dimensional manifold without
boundary and $\pi:\bb{P}\to X$ a principal $\n{U}(m)$-bundle equipped with a connection $\omega$. Assume that there is a global section $s:X\to \bb{P}$. Then, the quantity
$$
\rr{cs}(\omega)\;:=\;\int_Xs^*\f{CS}(\omega)\;\qquad\text{\rm mod.}\ \ \Z\;.
$$
is called the  \emph{Chern-Simons invariant}
$\rr{cs}(\omega) \in \R/\Z$ associated to $\omega$.
\end{definition} 
 
\medskip

\noindent
The following result shows that the Chern-Simons invariant is well defined.
\begin{proposition}
The Chern-Simons invariant does not dependent the choice of a particular 
 global section $s:X\to \bb{P}$, and depends only on the  equivalence class of $\omega$
up to gauge transformations.
\end{proposition}
\proof
Two global sections of $s_1$ and $s_2$ of $\bb{P}$
are related by a unique map $g:X\to \n{U}(m)$ such that $s_2(x)=s_1(x)\cdot g(x)$. Lemma \ref{lemma_CS_1}, the Stokes' theorem and the fact that $X$ has no boundary imply
$$
\int_X\Big(s^*_1\f{CS}(\omega)-s^*_2\f{CS}(\omega)\Big)\;=\;\int_X\Lambda(g)\;=:\;N_g\;\in\;\Z\;.
$$
The integer $N_g$  corresponds to the \virg{degree} of the map $g$. With a similar argument one can show that 
$\rr{cs}(\omega)=\rr{cs}(\omega')$ if $\omega$ and $\omega'$ are related by 
the transformation induced by an element of the gauge group.
\qed
 
\medskip

When the   principal $\n{U}(2m)$-bundle  $\pi:\bb{P}\to X$ is endowed with a Q-structure $\hat{\Theta}$
 it results  natural to use an equivariant Q-connection $\omega\in\rr{A}_{Q}(\bb{P})$ to define the Chern-Simons 3-form $\f{CS}(\omega)$. The  Q-structure  $\hat{\Theta}$ induces a symmetry of $\f{CS}(\omega)$.
\begin{lemma}\label{lemma:inv_CS_for}
Let $(\bb{P}, \hat{\Theta})$ be a $\n{U}(2m)$ Q-bundle over the involutive manifold $(X,\tau)$ which verifies Assumption \ref{ass:smooth}. Let $\omega\in\rr{A}_{Q}(\bb{P})$ be an equivariant connection and $\f{CS}(\omega)\in\Omega^3(\bb{P})$ the associated Chern-Simons 3-form. Then, the following equation
$$
\hat{\Theta}^*\f{CS}(\omega)\;=\; \f{CS}(\omega)
$$
 holds true.
\end{lemma}
\proof 
The equivariance of $\omega$ means that $\hat{\Theta}^*\omega=Q\overline{\omega}Q^{-1}=-Q{^{t}{\omega}}Q^{-1}$
where we used $\overline{\omega}=-{^{t}{\omega}}$ since the form  $\omega$ takes value in the Lie algebra $\rr{u}(2m)$.
The cyclicity of the trace provides
$$
\hat{\Theta}^*\f{CS}(\omega)\;=\; \f{CS}(\hat{\Theta}^*\omega)\;=\; \frac{1}{8\pi^2}\;{\rm Tr}\left({^{t}{\omega}}\wedge\dd{^{t}{\omega}}\;+\;\frac{2}{3}\;{^{t}{\omega}}\wedge{^{t}{\omega}}\wedge{^{t}{\omega}}\right)\;. 
$$
The identity ${^{t}{\omega}_1}\wedge {^{t}{\omega}_2}=(-1)^{q_1q_2}\;{^t(\omega_2\wedge\omega_1)}$ valid for each pair
 $\omega_1\in\Omega^{q_1}(\bb{P},\rr{u}(2m))$ and $\omega_2\in\Omega^{q_2}(\bb{P},\rr{u}(2m))$
and the invariance of the trace under the operation of taking the   \emph{transpose} imply
$$
\begin{aligned}
\hat{\Theta}^*\f{CS}(\omega)\;&=\; \frac{1}{8\pi^2}\;{\rm Tr}\left({\dd{\omega}}\wedge{{\omega}}\;+\;\frac{2}{3}{ {\omega}}\wedge{ {\omega}}\wedge{ {\omega}}\right)\\
&=\;\f{CS}(\omega)\;+\;\frac{1}{8\pi^2}\;{\rm Tr}\left({\dd{\omega}}\wedge{{\omega}}- {{\omega}}\wedge\dd{{\omega}}\right)\\
&=\;\f{CS}(\omega)\;+\;\frac{1}{8\pi^2}\;\dd{\rm Tr}\left({{\omega}}\wedge{{\omega}}\right)\;.
\end{aligned}
$$ 
To conclude the proof it is enough to observe that ${\rm Tr}\left({{\omega}}\wedge{{\omega}}\right)=0$ due to the anti-commutation relation of 1-forms.
\qed

\medskip

\noindent
The invariance of $\f{CS}(\omega)$ expressed in Lemma \ref{lemma:inv_CS_for} has an important  implication on the Chern-Simons invariant in low dimension, provided that certain conditions are met.
\begin{proposition}\label{propos:Z2_valuesCS_inv}
Let $(\bb{P}, \hat{\Theta})$ be a  $\n{U}(2m)$ Q-bundle over the involutive manifold $(X,\tau)$ which verifies Assumption \ref{ass:smooth}. Assume in addition that:
\begin{itemize}
\item[(a)] $X$ has dimension $d=3$ and $\tau$ reverses the orientation of $X$;
\vspace{1mm}
\item[(b)] There is a global section $s:X\to\bb{P}$ (not necessarily equivariant).
\end{itemize}
Then:
\begin{itemize}
\item[(i)] If  $\omega\in\rr{A}_{Q}(\bb{P})$ is an equivariant connection then
the associated Chern-Simons invariant $\rr{cs}(\omega)$ takes values in the set $\{0,\frac{1}{2}\}$;
\vspace{1.0mm}
\item[(ii)] $\rr{cs}(\omega)=\rr{cs}(\omega')$ for each pair of equivariant connections $\omega,\omega'\in\rr{A}_{Q}(\bb{P})$;
\vspace{1.0mm}
\item[(iii)] If $(\bb{P}, \hat{\Theta})$ admits a global equivariant section then $\rr{cs}(\omega)=0$ independently of $\omega\in\rr{A}_{Q}(\bb{P})$.
\end{itemize}
\end{proposition}
\proof
(i) Let $s:X\to\bb{P}$ be a global section. Since $\tau_\Theta(s):=\hat{\Theta}\circ s\circ \tau$ generally differs from $s$, there is a (unique) map $g:X\to\n{U}(2m)$ such that $\tau_\Theta(s)=s\cdot g$. Then
$$
\tau^*\big(s^*\f{CS}(\omega)\big)\;=\;(s\circ\tau)^*\f{CS}(\omega)\;=\; (s\cdot g)^*\big(\hat{\Theta}^*\f{CS}(\omega)\big)\;=\; (s\cdot g)^*\f{CS}(\omega)
$$
where in the last equality we used the result of Lemma \ref{lemma:inv_CS_for}. By exploiting the fact that $\tau$ 
reverses the orientation of $X$ one has
$$
\int_Xs^*\f{CS}(\omega)\;=\;-\int_X\tau^*\big(s^*\f{CS}(\omega)\big)\;=\;- \int_X(s\cdot g)^*\f{CS}(\omega)
\;=\;- \int_Xs^*\f{CS}(\omega)\;+\;N_g
$$
where $N_g:=\int_X\Lambda(g)\in\Z$. This implies that $2\rr{cs}(\omega)=0$, \ie $\rr{cs}(\omega)\in \{0,\frac{1}{2}\}$.\\
(ii) Let $\omega'$ be a second equivariant connection and consider the map $[0,1]\ni t\mapsto\omega_t:=(1-t)\omega +t\omega'\in\rr{A}_{Q}(\bb{P})$. Clearly $\rr{cs}(\omega_t)$ is a polynomial (hence continuous) function in $t$. On the other hand $\rr{cs}(\omega_t)\in \{0,\frac{1}{2}\}$ since $\omega_t$
 is equivariant. This implies that $\rr{cs}(\omega_{t_1})=\rr{cs}(\omega_{t_2})$ for all $t_1,t_2\in [0,1]$ and in particular $\rr{cs}(\omega)=\rr{cs}(\omega')$.\\
(iii) If $s$ is a global equivariant section one has
 $$
\tau^*\big(s^*\f{CS}(\omega)\big)\;=\;\tau^*\big(s^*\f{CS}(\omega)\big)\;=\;\tau^*\big(s^*\big(\hat{\Theta}^*\f{CS}(\omega)\big)\big)\;=\;\tau_\Theta(s)^*\f{CS}(\omega)\;=\;s^*\f{CS}(\omega)\;.
$$
Hence,
$$
\int_Xs^*\f{CS}(\omega)\;=\;\int_X\tau^*\big(s^*\f{CS}(\omega)\big)\;=\;-\int_Xs^*\f{CS}(\omega)
$$ 
which implies $\int_Xs^*\f{CS}(\omega)=0$.
\qed

\begin{remark}{\upshape
Due to the low dimensional assumption (a) in Proposition \ref{propos:Z2_valuesCS_inv} the assumption (b) about the existence of a global section is completely equivalent to the condition of vanishing of the first Chern class of the principal bundle. This condition is guaranteed by the stronger requirements: (1) $H^2_{\Z_2}(X,\Z(1))=0$, or (2) $H^2(X,\Z)=0$. }\hfill $\blacktriangleleft$
\end{remark}

\medskip

The following definition is justified by item (ii) of Proposition \ref{propos:Z2_valuesCS_inv} .
\begin{definition}[Intrinsic Chern-Simons invariant]\label{def:intrinsC-S_inv}
Let $(\bb{P}, \hat{\Theta})$ be a  $\n{U}(2m)$ Q-bundle over the involutive manifold $(X,\tau)$ such that $X$ has dimension $d=3$, $\tau$ reverses the orientation of $X$ and $\bb{P}$ admits a global section. Then the quantity
$$
\rr{cs}(\bb{P},\hat{\Theta})\;:=\;\rr{cs}(\omega)\;,\qquad\quad \text{for some}\ \ \omega\in \rr{A}_{Q}(\bb{P})
$$
does not depend on the choice of $\omega\in \rr{A}_{Q}(\bb{P})
$
and defines an \emph{intrinsic} (Chern-Simons) invariant for $(\bb{P},\hat{\Theta})$.
\end{definition}
\begin{remark}[A formula for the Chern-Simons invariant]
{\upshape
Let $(X, \tau)$ be a  three-dimensional involutive manifold   satisfying the assumption $H^2_{\Z_2}(X,\Z(1))=0$. As a consequence of Lemma \ref{lemma:prep1}
and the isomorphism \eqref{eq:iso_vect_princ_QR},
any $\n{U}(2m)$ Q-bundle $(\bb{P},\hat{\Theta})$ over $(X, \tau)$ can be represented by a smooth map $\xi : X \to \n{U}(2m)$ such that $\tau^*\xi = - Q \overline{\xi}^{-1} Q$. The average construction applied to the trivial connection on the product bundle \cite[Example 2.15]{denittis-gomi-15} gives an equivariant connection $\omega$, whose pull-back under the trivial section $s$ is $s^*\omega = \frac{1}{2} \sigma(\xi^{-1}d\xi)$. We then have $s^*\f{CS}(\omega) = \frac{1}{2}\Lambda(\xi)$, and hence the formula 
$$
\rr{cs}(\bb{P},\hat{\Theta}) = \frac{1}{2} \int_X \Lambda(\xi) \mod \Z\;.
$$
This formula can be compared with \cite[Proposition 11.21]{freed-moore-13}.
\hfill $\blacktriangleleft$}
\end{remark}

 %-----------------%
\subsection{Wess-Zumino term in absence of boundaries
}
\label{sec:Wess_Zumono_term}
In the last section we described the Chern-Simons invariant in the case of  three-dimensional base manifolds without boundary. In the case of manifolds with boundary  the Chern-Simons invariant itself depends on the choice of a section while the difference of the values of the Chern-Simons invariants depends only on the topological information on the boundary. This information is detected by the so-called \emph{Wess-Zumino term}. The latter is a topological quantity initially defined in the context of
of certain  two-dimensional conformal field theory known as
Wess-Zumino-Witten models. 
An excellent introduction to the theory of Wess-Zumino-Witten models is provided by the lecture notes
  \cite{gawedzki-99}. The presentation given here of the properties of the
   Wess-Zumino term 
 follows mainly \cite{freed-95}.

 \medskip

\begin{definition}[Wess-Zumino term] 
\label{dfn:WZ_without_boundary}
Let $\Sigma$ be a compact oriented manifold without boundary of dimension $d=2$. For any map $\xi : \Sigma \to \n{SU}(2)$, the \emph{Wess-Zumino term} $\f{WZ}_{\Sigma}(\xi) \in \R/\Z$ is defined by
$$
\f{WZ}_{\Sigma}(\xi)\; :=\;  
\int_{X_\Sigma} \Lambda\big(\widetilde{\xi}\big) \;\qquad\text{\rm mod.}\ \ \Z\;
$$
where 
$$
\Lambda\big(\widetilde{\xi}\big)\; :=\; -\frac{1}{24\pi^2} 
 {\rm Tr}\left(\widetilde{\xi}^{-1}\dd\widetilde{\xi}\right)^3 
 $$
according to the notation \eqref{notat_Lmab},
$X_\Sigma$ is any compact three-dimensional oriented  manifold whose boundary coincides with $\Sigma$, \ie
 $\partial X_\Sigma = \Sigma$, and $\widetilde{\xi} : X_\Sigma  \to \n{SU}(2)$ is any extension of $\xi$.
\end{definition}
\medskip

\noindent
Notice that the extended manifold $X_\Sigma$ and the extended section $\widetilde{\xi}$
in Definition \ref{dfn:WZ_without_boundary}
exist always.
The existence of $X_\Sigma$ follows from the vanishing of the second \emph{bordism group}\footnote{The existence of $X_\Sigma$ can be also justified by observing that closed oriented two-dimensional manifolds are classified by the genus and a genus $g$ surface is always the boundary of a  three-dimensional manifold. For instance the sphere $\n{S}^2$ is the boundary of the   three-dimensional disk $\n{D}^3$. Similarly
the torus $\n{T}^2$ is the boundary  of the manifold $\n{S}^1\times\n{D}^2$. The same occurs for higher genus surfaces.
},  $\Omega_2= 0$  \cite[Section 7]{milnor-stasheff-74}. 
The existence of $\widetilde{\xi}$ is due to $\pi_k(\n{SU}(2)) = 0$ for $k = 0, 1, 2$ plus a standard application of the Oka's (type) principle to pass from continuous  sections to smooth sections.
Finally, the condition  $\xi : \Sigma \to \n{SU}(2)$ can be relaxed by asking that the section 
$\xi : \Sigma \to \n{U}(2)$ possesses a determinant section ${\rm det}(\xi) : \Sigma \to \n{U}(1)$ which is null-homotopic. 

\medskip

The well-posedness of Definition \ref{dfn:WZ_without_boundary} is justified in the following result.

\begin{lemma}[Polyakov-Wiegmann formula] \label{lem:Polyakov_Wiegmann}
The Wess-Zumino term is independent of the choice of the  extensions $X_\Sigma$ and $\widetilde{\xi}$. Moreover, for every pair of sections $\xi_j : \Sigma \to \n{SU}(2)$, $j = 1, 2$, the \emph{Polyakov-Wiegmann formula} 
$$
\f{WZ}_{\Sigma}(\xi_1 \xi_2)= \f{WZ}_{\Sigma}(\xi_1) + \f{WZ}_{\Sigma}(\xi_2)
+ \frac{1}{8\pi^2}
\int_{\Sigma} \Tr(\xi_1^{-1}\dd\xi_1 \wedge \dd\xi_2 \xi_2^{-1})
$$
holds in $\R/\Z$.
\end{lemma}
\begin{proof} Given $\Sigma$ and  $\xi : \Sigma \to \n{SU}(2)$ as in Definition \ref{dfn:WZ_without_boundary} consider two extended manifolds $X_\Sigma$ and $X'_\Sigma$ such that 
$\partial X_\Sigma=\Sigma=\partial X'_\Sigma$ and two extended sections $\widetilde{\xi}$ and $\widetilde{\xi}'$ such that $\widetilde{\xi}|_\Sigma=\xi=\widetilde{\xi}'|_\Sigma$.
By reversing the orientation of $X'_\Sigma$ and then gluing it with $X_\Sigma$ along $\Sigma$ one obtains a 
compact oriented three-dimensional manifold $X:=(-X'_\Sigma)\sqcup X_\Sigma$, where
the minus sign indicates the reversal of the orientation. Similarly, $\widetilde{\xi}$ and $\widetilde{\xi}'$
can be glued together to define a section $\xi_X:=(\widetilde{\xi}\sqcup\widetilde{\xi}'):X\to \n{SU}(2)$.
 It is well-known that 
 $$
 \int_X\Lambda\left({\xi}_M\right)\;=\;-\frac{1}{24\pi^2} \int_X
 {\rm Tr}\left({\xi}^{-1}_X\dd{\xi}_X\right)^{\wedge 3} \;\in\;\Z\;.
 $$
On the other hand, one has that
$$
\int_X\Lambda\left({\xi}_X\right)\;=\;\int_{X_\Sigma}\Lambda\big(\widetilde{\xi}\big)\;-\;\int_{X'_\Sigma}\Lambda\big(\widetilde{\xi}'\big)\;\in\;\Z
$$
where the minus sign is justified by the inversion of the orientation. Thus, since the Wess-Zumino term $\f{WZ}_{\Sigma}(\xi)$ is defined modulo an integer, it can be computed equivalently through the pair $X_\Sigma, \widetilde{\xi}$ or the pair $X'_\Sigma, \widetilde{\xi}'$. The Polyakov-Wiegmann formula for  
$\f{WZ}_{\Sigma}(\xi_1 \xi_2)$ follows from an explicit computation. By taking extensions of $\xi_1$ and $\xi_2$ one computes $\Lambda(\xi_1\xi_2) - \Lambda(\xi_1) - \Lambda(\xi_2)$ directly. Then, the 
integration over $X_\Sigma$ and the application of the 
Stokes' theorem to obtain the integral on the boundary $\Sigma$   provide the final result. 
\end{proof}

\medskip

From  formula \eqref{eq:formul_CS_W} and the Stokes' theorem one immediately deduces the following result:
\begin{lemma}\label{lemma_WZ_1}
Let $X$
 be a compact oriented manifold 
 of dimension $d=3$  with non-empty boundary $\Sigma:=\partial X$.
Let $\pi:\bb{P}\to X$ be a principal $\n{U}(2)$-bundle equipped with a connection $\omega$
and a global (smooth) section $s:X\to\bb{P}$.
Let  $g:X\to\n{U}(2)$ be any 
(smooth) map such that ${\rm det}(g):X\to\n{U}(1)$ is null-homotopic. Then the following formula 
$$
\int_Xs^*_g\f{CS}(\omega)\;-\; \int_Xs^*\f{CS}(\omega)\;=\;-\frac{1}{8\pi^2}\int_\Sigma{\rm Tr}\big(s^*\omega\wedge\dd g^{-1} g\big)\;+\;\f{WZ}_{\Sigma}\left(g|_{\Sigma}\right)\;\qquad\text{\rm mod.}\ \ \Z
$$
holds true.
\end{lemma}

 %-----------------%
\subsection{Wess-Zumino term in presence of boundaries
}
\label{sec:Wess_Zumono_term_no-bound}

In the continuation of this work we will 
 be interested in calculating  the Wess-Zumino term through \virg{cutting and pasting}. To setup the machinery, we need to extend the definition of the Wess-Zumino term for two-dimensional manifolds with boundary. To do that let us observe that associated to a compact oriented one-dimensional manifold $S$ without boundary (union of circles), there exists a Hermitian line bundle $\mu:\bb{L}_S\to {\rm Map}(S,\n{SU}(2))$. 
The specific structure of this line bundle will be not used in this work and for this reason 
the details of the construction of  $\bb{L}_S$ 
 will be only sketched .
 The interested reader can refer to 
  \cite[Appendix A]{freed-95} or to \cite[Section 1.3]{kohno-02} for a  more rigorous presentation.
  
 \medskip
 
 Given $S$ consider a  two-dimensional manifolds $D_S$ with boundary $\partial D_S=S$ along with the space ${\rm Map}(D_S,\n{SU}(2))$. Given an element  $\widetilde{\gamma}\in {\rm Map}(D_S,\n{SU}(2))$ its restriction,  denoted with $\gamma:=\widetilde{\gamma}|_S$,  defines an element in ${\rm Map}(S,\n{SU}(2))$. Let $\widetilde{\gamma}_1,\widetilde{\gamma}_2\in {\rm Map}(D_S,\n{SU}(2))$ 
two maps which agree on the boundary $S$, namely 
  such that $\gamma_1=\gamma_2$. Such two maps can be glued together to produce a map
 $\xi_{(1,2)}:=\widetilde{\gamma}_1\sqcup \widetilde{\gamma}_2$
on the two-dimensional manifolds without boundary
$\Sigma_S:=(- D_S)\sqcup D_S$ obtained by gluing two copies of $D_S$ (with opposite orientation) along the common boundary.
As a consequence  the quantity $\f{WZ}_{\Sigma_S}(\xi_{(1,2)})$ turns out to be well defined according to Definition \ref{dfn:WZ_without_boundary}.
Consider now the space
$$
\bb{L}_S\;:=\;\big({\rm Map}(D_S,\n{SU}(2))\;\times\;\C\big)\;/\;{\sim}
$$  
where the equivalence relation 
$\sim$ is defined as follows: Let $\widetilde{\gamma}_1, \widetilde{\gamma}_2\in {\rm Map}(D_S,\n{SU}(2))$ and $z_1,z_2\in\C$ then 
$$
(\widetilde{\gamma}_1,z_1)\;\sim\; (\widetilde{\gamma}_2,z_2)\;\qquad\Leftrightarrow\qquad \gamma_1\;=\;\gamma_2\;,\quad z_1\;=\;z_2\;\expo{\ii2\pi\;\f{WZ}_{\Sigma_S}(\xi_{(1,2)})}\;.
$$ 
The space $\bb{L}_S$ defined in this way turns out to be the total space of a complex line bundle over ${\rm Map}(S,\n{SU}(2))$
with projection $\mu:\bb{L}_S\to {\rm Map}(S,\n{SU}(2))$ given by 
$$\mu\;:\;\big[\widetilde{\gamma},z\big]\;\longmapsto\;\gamma\;:=\;\widetilde{\gamma}|_{S}\;
$$
where $\gamma:=\widetilde{\gamma}|_{S}$ is independent of the choice of the representative by construction.

  \medskip
   
   Henceforth,  only  the following properties  of the line boundle $\mu:\bb{L}_S\to {\rm Map}(S,\n{SU}(2))$ will be relevant \cite[Proposition A.1]{freed-95}: 
\begin{itemize}
\item[(i)] For $\gamma_1,\gamma_2\in{\rm Map}(S,\n{SU}(2))$ let $\gamma_1\gamma_2\in{\rm Map}(S,\n{SU}(2))$ defined by the pointwise multiplication. Then there is an isometry 
\begin{equation}\label{eq:prod_fib}
\mu^{-1}(\gamma_1)\;\otimes\;\mu^{-1}(\gamma_2)\;\longrightarrow\;\mu^{-1}(\gamma_1\gamma_2)
\end{equation}
which involves the fibers of $\bb{L}_S$ over $\gamma_1$, $\gamma_2$ and $\gamma_1\gamma_2$; 
\vspace{1mm}
\item[(ii)] The \emph{product of fibers} \eqref{eq:prod_fib} defined by the isometry is associative;
\vspace{1mm}
\item[(iii)] If $\gamma_0\in {\rm Map}(S,\n{SU}(2))$ is the constant map
then there is a trivialization $\mu^{-1}(\gamma_0)\simeq\C$ which respect  \eqref{eq:prod_fib}.
\end{itemize}

 \medskip
 
All the ingredients are now available for extending the Definition \ref{dfn:WZ_without_boundary} to manifolds with 
 boundary.
 
\begin{definition}[Wess-Zumino term with boundary]
\label{dfn:WZ_with_boundary}
Let $\Sigma$ be a compact oriented manifold  of dimension $d=2$ with one-dimensional (compact and oriented) boundary  $S:=\partial\Sigma$. Let $\mu:\bb{L}_{S} \to {\rm Map}(S,\n{SU}(2))$
be the associated line bundle. Every $\xi:\Sigma\to \n{SU}(2)$ gives rise to a point 
$\xi|_{S}\in {\rm Map}(S,\n{SU}(2))$ and an associated fiber 
$\mu^{-1}(\xi|_{S})\subset \bb{L}_{S}$. 
Let $D_S$ be a disk (contractible two dimensional manifold) with boundary $\partial D_S =S=\partial\Sigma$. Given any $\zeta_{D_S}:D_s\to\n{SU}(2)$ such that $\zeta_{D_S}|_S=\xi_S$
let $\xi\sqcup \zeta_{D_S}$ be the map defined on the closed manifold $\Sigma_D:=\Sigma\sqcup(-D_S)$ by the gluing of the functions $\zeta_{D_S}$ and $\xi$ along the common boundary $S$.
The Wess-Zumino term $\f{WZ}_{\Sigma}(\xi)$ is then defined by the following equation
$$
\expo{\ii2\pi\;\f{WZ}_{\Sigma}(\xi)}\;:=\;\left[\zeta_{D_S},\expo{\ii\;2\pi \f{WZ}_{\Sigma_D}(\xi\sqcup \zeta_{D_S})}\right]\;\in\;\mu^{-1}(\xi|_{S})\;.
$$
\end{definition}

\medskip

To introduce the next result it is worth mentioning that given a complex vector bundle $\bb{E}\to X$ its \emph{conjugate} $\overline{\bb{E}}\to X$ is the complex vector bundle  whose underlying total space agrees with $\bb{E}$ as a set, but with inverted complex structure with respect to the multiplication by scalars $z \in \C$.
 If $\bb{E}$ is endowed with a Hermitian metric, then so is $\overline{\bb{E}}$. This allows the identification of $\overline{\bb{E}}$ with the \emph{dual} vector bundle $\bb{E}^*$.

\begin{proposition}[Orientation]\label{prop:orient}
The following facts hold true:
\begin{itemize}
\item[(i)]
Let $S$ be a compact oriented one-dimensional manifold without boundary, and $-S$ the same manifold with  reversed orientation. Then there exists a natural isometric isomorphism
$$
\bb{L}_{-S}\; \simeq\; \overline{\bb{L}_S}\;.
$$
\vspace{1mm}

\item[(ii)]
Let $\Sigma$ be a compact oriented two-dimensional manifold with boundary, and $- \Sigma$ the same manifold  with reversed orientation. For any $\xi : \Sigma \to \n{SU}(2)$ the relation
$$
\f{WZ}_{-\Sigma}(\xi)\;=\;-{\f{WZ}_{\Sigma}(\xi)}
$$
holds true.
\end{itemize}

\end{proposition}
\medskip

\noindent
Property (i) of Proposition \ref{prop:orient} is a direct consequence of the construction of the space $\bb{L}_S$. Property (ii) follows from Definition \ref{dfn:WZ_with_boundary} 
under the isometry described in (i).

\medskip

\begin{remark}[Central extension of the loop group]
\label{rk:cent_ext}
{\upshape
Definition \ref{dfn:WZ_with_boundary} will be mainly applied to two-dimensional manifolds $\Sigma$
such that
$\partial\Sigma\simeq\n{S}^1$. In this case we will  write $\bb{L}_{\n{S}^1}$ instead of  $\bb{L}_{\partial\Sigma}$.
The set ${\rm Map}(\n{S}^1,\n{SU}(2))$ 
endowed with the pointwise multiplication is known as the \emph{loop group}  of $\n{SU}(2)$ \cite{pressley-segal-86}, and will be denoted here with   ${\rm Loop}_{\n{SU}(2)}$. 
The total space $S(\bb{L}_{\n{S}^1})$ of the principal $\n{U}(1)$-bundle (also known as \emph{circle-bundle}) associated to $\bb{L}_{\n{S}^1}$ inherits a group structure from the  product of fiber \eqref{eq:prod_fib}. 
This gives rise to a central extension of ${\rm Loop}_{\n{SU}(2)}$:
$$
1\;\longrightarrow\;\n{U}(1) \;\longrightarrow\; S(\bb{L}_{\n{S}^1}) \;\longrightarrow\; {\rm Loop}_{\n{SU}(2)}\;\longrightarrow\; 1\;.
$$
Let $\xi_0:\Sigma\to\n{SU}(2)$ be the constant map with value the identity matrix $\n{1}_{\C^2}\in \n{SU}(2)$.
By construction of the product of fiber \eqref{eq:prod_fib} one has that $[\xi_{0},\expo{\ii\;2\pi \f{WZ}_{\Sigma_D}(\xi_0\sqcup \xi_0)}]$ acts as the unit of the group $S(\bb{L}_{\n{S}^1})$. Therefore, by invoking Definition
\ref{dfn:WZ_with_boundary} one obtains that  $\expo{\ii2\pi\;\f{WZ}_{\Sigma}(\xi_0)}\in \bb{L}_{\n{S}^1}$ provides the unit of the central extension $S(\bb{L}_{\n{S}^1})$.
For a more complete description of this central extension the reader is referred to \cite{pressley-segal-86,freed-95,kohno-02}.
}\hfill $\blacktriangleleft$
\end{remark}

\medskip

The link between Definition \ref{dfn:WZ_without_boundary} and Definition \ref{dfn:WZ_with_boundary} is provided by the following result.
\begin{proposition}[Gluing property]
\label{prop:gluing_prop}
Let $\Sigma$ be a compact oriented two-dimensional manifold without boundary. Assume that  $\Sigma$ can be cut along  an embedded circle $\n{S}^1$ to get two compact oriented two-dimensional manifolds $\Sigma_1$ and $\Sigma_2$ such that $\partial \Sigma_1 \simeq -\n{S}^1$, $\Sigma_2 \simeq \n{S}^1$ in such a way that $\Sigma = \Sigma_1 \sqcup \Sigma_2$. Then, for any $\xi : \Sigma \to \n{SU}(2)$ it holds that
\begin{equation}\label{eq:duality_WZ}
\expo{\ii2\pi\;\f{WZ}_{\Sigma}(\xi)}\;=\;\langle\expo{\ii2\pi\;\f{WZ}_{\Sigma_1}(\xi|_{\Sigma_1})};\expo{\ii2\pi\;\f{WZ}_{\Sigma_2}(\xi|_{\Sigma_1})}\rangle
\end{equation}
where $\langle\;;\;\rangle$ denotes the contraction between $\expo{\ii2\pi\;\f{WZ}_{\Sigma_1}(\xi|_{\Sigma_1})}\in \bb{L}_{\n{S}^1}$ and $\expo{\ii2\pi\;\f{WZ}_{\Sigma_2}(\xi|_{\Sigma_2})}\in \bb{L}_{\n{S}^1}^*$.
\end{proposition}
\medskip

\noindent
Equation \eqref{eq:duality_WZ} can be reformulated in the suggestive formula
$$
\f{WZ}_{\Sigma}(\xi)\;=\;\f{WZ}_{\Sigma_1}(\xi|_{\Sigma_1})\;-\;\f{WZ}_{\Sigma_2}(\xi|_{\Sigma_2}) \;\qquad\text{\rm mod.}\ \ \Z\;.
$$
A proof of a generalized version of  Proposition \ref{prop:gluing_prop} can be found in \cite[Section 1.3]{kohno-02}. 

\medskip

Although  simplified, the  version of the gluing property described in Proposition \ref{prop:gluing_prop}
 is sufficient for the purposes of this work.
Indeed, the  gluing property will be mainly applied to the  situation described below:

\begin{remark}
\label{rk:special_gluing}
{\upshape
 Let $\Sigma_1$ and $\Sigma_2$ be compact oriented two-dimensional manifolds without boundary. Assume that an embedded disk ${D}$ can be cut  out from both the manifolds in such a way that $\Sigma_1 = \Sigma'_1 \sqcup {D}$ and 
 $\Sigma_2 = \Sigma'_2 \sqcup {D}$ where 
 $\Sigma'_1$ and $\Sigma'_2$ are 
 two-dimensional manifolds with boundaries
 $\partial \Sigma_1 \simeq \partial \Sigma_2\simeq-\partial {D} \simeq -\n{S}^1$. 
 Let $\xi_1 : \Sigma_1 \to \n{SU}(2)$ and
 $\xi_2 : \Sigma_2 \to \n{SU}(2)$
  be two maps such that $\xi_1|_{{D}} = \xi_2|_{{D}}$ and both $\xi_1$ and $\xi_2$ have constant value $\n{1}_{\C^2}$ on a neighborhood of $\Sigma'_1\subset\Sigma_1$ and $\Sigma'_2\subset\Sigma_2$, respectively. Under this setting  it holds that
\begin{equation}
\label{eq:main_tech_res}
  \f{WZ}_{\Sigma_1}(\xi_1)\;=\;  \f{WZ}_{\Sigma_2}(\xi_2)\;\qquad\text{\rm mod.}\ \ \Z\;.
\end{equation}
In fact both $\expo{\ii2\pi\;\f{WZ}_{\Sigma'_1}(\xi_1|_{\Sigma'_1})}\in \bb{L}_{\n{S}^1}^*$ and $\expo{\ii2\pi\;\f{WZ}_{\Sigma'_2}(\xi_2|_{\Sigma'_2})}\in \bb{L}_{\n{S}^1}^*$ describe  the unit of the central extension $S(\bb{L}_{\n{S}^1})$ as discussed in Remark \ref{rk:cent_ext}. Therefore, 
$$
\expo{\ii2\pi\;\f{WZ}_{\Sigma'_1}(\xi_1|_{\Sigma'_1})}\;=\;\expo{\ii2\pi\;\f{WZ}_{\Sigma'_2}(\xi_2|_{\Sigma'_2})}\;,\qquad \expo{\ii2\pi\;\f{WZ}_{\n{D}}(\xi_1|_{\n{D}})}\;=\;\expo{\ii2\pi\;\f{WZ}_{\n{D}}(\xi_2|_{\n{D}})}\;
$$ 
where the second equality  follows from the assumption $\xi_1|_{{D}} = \xi_2|_{{D}}$.
 By applying the gluing property \eqref{eq:duality_WZ} one  gets
 $\expo{\ii2\pi\;\f{WZ}_{\Sigma_1}(\xi_1)}=\expo{\ii2\pi\;\f{WZ}_{\Sigma_2}(\xi_2)}$ which justifies equation \eqref{eq:main_tech_res}.
 }\hfill $\blacktriangleleft$
\end{remark}

 %-----------------%
\subsection{Classification via Wess-Zumino term in dimension two
}
\label{sec:Wess_Zumono_term_class_D2}
In this section the description of rank 2 Q-bundles over an oriented two-dimensional FKMM-manifold $(\Sigma,\tau)$ obtained in Section \ref{subsec:altern_present} and Section \ref{subsec:altern_FKMM-2D} will be combined with the theory of the Wess-Zumino term described in Section \ref{sec:Wess_Zumono_term} and Section \ref{sec:Wess_Zumono_term_no-bound} in order to prove that the Wess-Zumino term completely classifies ${\rm Vec}^{2}_{Q}(\Sigma, \tau)$.

\medskip

The following three preliminary results are needed.
\begin{lemma} \label{lem:WZ_for_equivariant_map}
Let $(\Sigma,\tau)$ be an oriented two-dimensional FKMM-manifold in the sense of Definition \ref{def:good_manif_d=2}.
Let ${\rm Map}(\Sigma, \n{SU}(2))_{\Z_2}$ be the set of equivariant maps described by \eqref{eq:coxxi-01} and $[\Sigma, \n{SU}(2)]_{\Z_2}$ the set of equivalence classes under the $\Z_2$-homotopy equivalence. The following facts hold true:
\begin{itemize}
\item[(i)]
The exponentiated Wess-Zumino term of $\xi \in {\rm Map}(\Sigma, \n{SU}(2))_{\Z_2}$ takes values in $\Z_2$, so that one gets a map
$$
\expo{\ii\; 2\pi \f{WZ}_{\Sigma}}\; :\; {\rm Map}(\Sigma, \n{SU}(2))_{\Z_2}\; \longrightarrow\; \Z_2\;;
$$
\item[(ii)]
The map above is invariant under the $\Z_2$-homotopy, and hence induces a map
$$
\expo{\ii\; 2\pi \f{WZ}_{\Sigma}}\; :\; [\Sigma, \n{SU}(2)]_{\Z_2}\; \longrightarrow\; \Z_2\;.
$$
\end{itemize}
\end{lemma}
\proof (i) For every  $\xi\in {\rm Map}(\Sigma, \n{SU}(2))_{\Z_2}$ the quantity $\f{WZ}_{\Sigma}(\xi)\in\R/\Z$ is defined according to
Definition
\ref{dfn:WZ_without_boundary}.
Since $\xi$ satisfies $\tau^* \xi = \xi^{-1}$, the diffeo-invariance (functoriality) of the Wess-Zumino term \cite{freed-95} implies
$$
\f{WZ}_{\Sigma}(\xi)\; =\; \f{WZ}_{\Sigma}(\tau^*\xi)\; =\; \f{WZ}_{\Sigma}(\xi^{-1})\;.
$$
Form the relation $\zeta^{-1}\dd \zeta  = - \zeta \dd \zeta^{-1}$, valid for generic map with values in $\n{SU}(2)$ if follows that ${\rm Tr}(\zeta^{-1}\dd \zeta)^{n}=(-1)^n{\rm Tr}(\zeta\dd \zeta^{-1})^{n}$. The application of this identity to the Wess-Zumino term implies $\f{WZ}_{\Sigma}(\xi^{-1})=-\f{WZ}_{\Sigma}(\xi)$. In conclusion one obtains that $\f{WZ}_{\Sigma}(\xi)=-\f{WZ}_{\Sigma}(\xi)$ modulo $\Z$, \ie $2\f{WZ}_{\Sigma}(\xi)\in\{0,1\}$. This proves that the exponential map in (i) takes values in $\Z_2$.\\
(ii) If $\widehat{\xi} : \Sigma \times [0, 1] \to \n{SU}(2)$ is a $\Z_2$-homotopy, then the map
$$
[0,1]\;\ni\;t\;\longmapsto\;\f{WZ}_{\Sigma}(\widehat{\xi}|_{\Sigma\times\{t\}})\in\R/\Z
$$
is continuous. Hence, the value of the exponential $\expo{\ii\; 2\pi \f{WZ}_{\Sigma}}(\widehat{\xi}|_{\Sigma\times\{t\}})$ must be constant for all $t$ in view of the discreteness of the target space. This concludes the proof.
\qed

\begin{lemma} \label{lem:surjection_with_WZ}
Let $(\Sigma,\tau)$ be an oriented two-dimensional FKMM-manifold in the sense of Definition \ref{def:good_manif_d=2}. For each $\epsilon\in {\rm Map}(\Sigma^\tau,\{\pm1\})$ there exists a $\xi_\epsilon\in{\rm Map}(\Sigma, \n{SU}(2))_{\Z_2}$ such that $\Phi_\kappa(\xi_\epsilon)=\epsilon$ and
$$
\expo{\ii\; 2\pi \f{WZ}_{\Sigma}(\xi_\epsilon)}
\;=\;\Pi(\epsilon)
$$
where the map $\Pi$ is defined by \eqref{eq:Pi_map}.
\end{lemma}
\proof
The proof of Lemma \ref{lemm:iso_for_class} contains the recipe  to construct a map $\xi_\epsilon\in{\rm Map}(\Sigma, \n{SU}(2))_{\Z_2}$ for each $\epsilon\in {\rm Map}(\Sigma^\tau,\{\pm1\})$ such that $\Phi_\kappa(\xi_\epsilon)=\epsilon$. Let $\Sigma^\tau=\{x_1,\ldots,x_n\}$ be a labeling for the fixed point set. For each label let $\epsilon_i\in {\rm Map}(\Sigma^\tau,\{\pm1\})$
be defined by $
\epsilon_i(x_j)=1-2\delta_{ij}$. Let $\xi_i:=\xi_{\epsilon_i}$ be the element in ${\rm Map}(\Sigma, \n{SU}(2))_{\Z_2}$ such that 
$\Phi_\kappa(\xi_i)=\epsilon_i$. Then, by construction, each  $\xi_\epsilon$ can be expressed by the pointwise product of a certain number of $\xi_i$. Let assume that $\xi_\epsilon=\xi_{i_1}\cdot\ldots\cdot\xi_{i_k}$. 
 Since the supports of the differential forms $\xi_i^{-1}\dd\xi_i$ are pairwise disjoint, 
  the Polyakov-Wiegmann formula (see Lemma \ref {lem:Polyakov_Wiegmann})
  provides 
$$
\f{WZ}_{\Sigma}(\xi_\epsilon) 
\;=\; \f{WZ}_{\Sigma}(\xi_{i_1})\;+\; \ldots \;+\; \f{WZ}_{\Sigma}(\xi_{i_k})\;\qquad\text{\rm mod.}\ \ \Z\;.
$$
The next task is to evaluate the generic term $\f{WZ}_{\Sigma}(\xi_{i})$. For that the construction in Remark \ref{rk:special_gluing} will be applied. Given $x_i\in\Sigma^\tau$ consider a small disk $D_i\subset\Sigma$ such that $\tau(D_i)=D_i$ and $x_i\in D_i$ is the only fixed point. The restriction $\xi_i|_{D_i}$ has by construction the following property: $\xi_i|_{D_i}(x_i)=-\n{1}_{\C^2}$ and $\xi_i|_{D_i}(x)=+\n{1}_{\C^2}$ if $x\in\partial D_i$.
By an equivariant diffeomorphism $D_i$ can be identified with the closed unit disk $D\subset\C$ endowed  with the involution $z \mapsto -z$ and the map 
$\xi_i|_{D_i}$ can be identified with the map $\xi_D$ described in the proof of Lemma \ref{lemm:iso_for_class}. By gluing two copies  $D$ and $D'$ of the same disk along the common boundary $\n{S}^1$ one obtains that $D\sqcup D'$ is identifiable  with the equivariant sphere $\n{S}^2$ with involution $(k_0,k_1,k_2)\mapsto(k_0,-k_1,-k_2)$
which fixes only the two poles $(\pm1,0,0)$.
Moreover, given the constant map $\xi_0:D'\to\n{1}_{\C^2}$, one has that the gluing $\xi_D\sqcup \xi_0$ identifies an equivariant map $\chi:\n{S}^2\to\n{SU}(2)$ such that $\chi(\pm1,0,0)=\pm \n{1}_{\C^2}$. Since the condition described in Remark \ref{rk:special_gluing} are met one has that
$$
\f{WZ}_{\Sigma}(\xi_{i})\;=\;\f{WZ}_{\n{S}^2}(\chi)\;\qquad\text{\rm mod.}\ \ \Z\;.
$$
A possible realization for $\chi$ is the following:
\begin{equation}
\label{eq:prefer_chi}
\chi(k_0,k_1,k_2)\;=\;\left(\begin{array}{cc}k_0 & -k_1+\ii k_2 \\k_1+\ii k_2 & k_0\end{array}\right)\;.
\end{equation}
Recall that $[\n{S}^2,\n{U}(1)]_{\Z_2}\simeq H^1_{\Z^2}(\n{S}^2,\Z(1))\simeq \Z_2$ is made by constant maps \cite[Proposition A.2]{gomi-13}. Then, the isomorphism $[\n{S}^2,\n{SU}(2)]_{\Z_2}\;/\;[\n{S}^2,\n{U}(1)]_{\Z_2}\simeq\Z_2$ obtained from Proposition \ref{prop:classification}
assures that, up to  a $\Z_2$-homotopy if necessary, one can  always choose the equivariant map $\chi$ as given in  \eqref{eq:prefer_chi}. The computation of 
$\f{WZ}_{\n{S}^2}(\chi)$ with $\chi$ given by 
\eqref{eq:prefer_chi} proceed as follows:
Consider the map $\widetilde{\chi}:\n{S}^3\to\ \n{SU}(2)$ defined by
\begin{equation}
\label{eq:prefer_chi_ext}
\widetilde{\chi}(k_0,k_1,k_2,k_3)\;=\;\left(\begin{array}{cc}k_0+\ii k_3 & -k_1+\ii k_2 \\k_1+\ii k_2 & k_0-\ii k_3\end{array}\right)\;.
\end{equation}
Let $\n{S}^3_+:=\{k\in\n{S}^3\ |\ k_3\geqslant0\}$ be the upper hemisphere. Then $\partial \n{S}^3_+\simeq\n{S}^2$ and $\widetilde{\chi}|_{\partial \n{S}^3_+}=\chi$.
Since $\n{S}^3_+$ is just half sphere one gets by a direct computation that
$$
\f{WZ}_{\n{S}^2}(\chi)\;=\;
 \frac{-1}{48\pi^2} \int_{\n{S}^3_+} {\rm Tr} \big(\widetilde{\chi}^{-1}\dd \widetilde{\chi}\big)^3\; =\; \frac{1}{2}\;.
$$
As a consequence $\expo{\ii\; 2\pi \f{WZ}_{\Sigma}(\xi_i)}=\expo{\ii\; 2\pi \f{WZ}_{\n{S}^2}(\chi)}=-1$ and 
$$
\expo{\ii\; 2\pi \f{WZ}_{\Sigma}(\xi_\epsilon)}\;
=\; \prod_{x_{i_1}, \ldots, x_{i_k}} (-1)
\;=\; \Pi(\epsilon)\;.
$$
This complete the proof.\qed

\begin{lemma}
Let $(\Sigma,\tau)$ be an oriented two-dimensional FKMM-manifold in the sense of Definition \ref{def:good_manif_d=2}. The Wess-Zumino term induces a well-defined map
$$
\expo{\ii\; 2\pi \f{WZ}_{\Sigma}}\; :\; 
[\Sigma, \n{SU}(2)]_{\Z_2}\;/\;[\Sigma, \n{U}(1)]_{\Z_2}\; \longrightarrow\; \Z_2\;.
$$
\end{lemma}
\proof
The lemma is proved if one can show that for any  $\xi \in {\rm Map}(\Sigma, \n{SU}(2))_{\Z_2}$ and $\phi  \in {\rm Map}(\Sigma, \n{U}(1))_{\Z_2}$ it holds that $\expo{\ii\; 2\pi \f{WZ}_{\Sigma}(\xi)}=\expo{\ii\; 2\pi \f{WZ}_{\Sigma}(\xi')}$ where
$$
\xi' = 
\left(
\begin{array}{cc}
\tau^*\phi & 0 \\
0 & 1
\end{array}
\right)\cdot
\xi
\cdot
\left(
\begin{array}{cc}
1 & 0 \\
0 & \phi
\end{array}
\right)\;.
$$
Let $\epsilon:=\Phi_\kappa(\xi)$ and $\epsilon':=\Phi_\kappa(\xi')$. Associated with the maps $\xi,\xi'\in {\rm Map}(\Sigma^\tau, \Z_2)$ one can construct the  associated maps  $\xi_\epsilon, \xi_{\epsilon'}\in {\rm Map}(\Sigma, \n{SU}(2))_{\Z_2}$ according to  Lemma \ref{lem:surjection_with_WZ}. Lemma \ref{lemm:iso_for_class} assures that $\xi$ and $\xi'$ are $\Z_2$-homotopy equivalent to $\xi_\epsilon$ and  $\xi_{\epsilon'}$, respectively. Thus
$$
\expo{\ii\; 2\pi \f{WZ}_{\Sigma}(\xi)}\;=\;
\expo{\ii\; 2\pi \f{WZ}_{\Sigma}(\xi_\epsilon)}\;=\;\Pi(\epsilon)\;=\;\Pi(\Phi_\kappa(\xi))
$$
and similarly for $\expo{\ii\; 2\pi \f{WZ}_{\Sigma}(\xi')}\;=\;\Pi(\Phi_\kappa(\xi'))
$.
Since Proposition \ref{prop:classification} assures that $\Phi_\kappa(\xi)=\Phi_\kappa(\xi')$ it follows that $\expo{\ii\; 2\pi \f{WZ}_{\Sigma}(\xi)}=\expo{\ii\; 2\pi \f{WZ}_{\Sigma}(\xi')}$. This completes the proof.
\qed

\medskip

We are now in position to prove the first main result of this work.

\proof[{Proof of Theorem \ref{teo:main1}}]
In view of the isomorphism proved in Theorem \ref{theo:main-SU2_charat} and the  resulting equality \eqref{eq:FKMM_from_map} it is enough
to show that 
 $\expo{\ii\; 2\pi \f{WZ}_{\Sigma}}=\Pi\circ \Phi_\kappa$ as maps
form $[\Sigma, \n{SU}(2)]_{\Z_2}/[\Sigma, \n{U}(1)]_{\Z_2}$ into $\Z_2$. 
From
Proposition \ref{prop:classification} and Theorem \ref{prop:Fu-Kane-Mele_formula1}
one gets that
 $\Pi\circ \Phi_\kappa$ is a bijection.
 Thus, it is enough to prove the equality 
 $\expo{\ii\; 2\pi \f{WZ}_{\Sigma}}=\Pi\circ \Phi_\kappa$ on ${\rm Map}(\Sigma, \n{SU}(2))_{\Z_2}$. However, this is clear from Lemma \ref{lem:surjection_with_WZ}.
\qed

\medskip

\noindent 
By using the arguments in Remark \ref{rk:higer-rk}, Theorem \ref{teo:main1} can be immediately generalized to the case of Q-bundles of rank $2m$.

%-----------------%
\subsection{Classification via Chern-Simons invariant in dimension three 
}
\label{sect:class_d=3}
The main aim of this section is to 
provide the proof of Theorem \ref{prop:Fu-Kane-Mele=chern_d=3}.
This proof is facilitated by a 
 particular presentation of 
 principal Q-bundles over $(X,\tau)$.
Suppose that $X^\tau = \{ x_1, \ldots, x_n \}$ consists of $n$ points. Thanks to the 
\emph{slice theorem} \cite[Chapter I, Section 3]{hsiang-75} for each $i = 1, \ldots, n$ one can find  a closed $\tau$-invariant disk $D_i$ centered at $x_i$ such that $D_i \cap D_j = \emptyset$ for $i \neq j$ and each $D_i$ is equivariantly diffeomorphic to the standard   unit disk in $\R^3$ with antipodal involution $\tau(x) = -x$. Define
$$
X_D \;:=\: \bigsqcup_{i=1,\ldots,n} D_i\;, \qquad\quad
X' \;:=\; X \backslash {\rm Int}(X_D)\;,
$$
so that $X = X' \sqcup X_D$. Given any map $\varphi : X' \cap D \to \n{U}(2)$ one can glue together the product bundles over $X'$ and $X_D$
to form a principal $\n{U}(2)$-bundle over $X$:
\begin{equation}\label{eq:specila_princ}
\bb{P}_\varphi\;:=\;(X' \times \n{U}(2))\; \bigsqcup_{\varphi}\; (X_D \times \n{U}(2))\;.
\end{equation}
Assume that  $\varphi\in{\rm Map}(X,\n{U}(2))_{\Z_2}$, namely $\varphi$ is equivariant with respect to the involution $\tau^*\varphi=-Q\overline{\varphi} Q$, then the principal $\n{U}(2)$-bundle $\bb{P}_\varphi$
gives rise to a principal Q-bundle. 
\begin{lemma}
\label{lemma:prepar_CS1}
Assume that the hypotheses of Theorem \ref{prop:Fu-Kane-Mele=chern_d=3} are met.
Any principal $\n{U}(2)$ Q-bundle $(\bb{P},\hat{\Theta})$ over $(X,\tau)$
 is isomorphic to a  principal $\n{U}(2)$ Q-bundle $\bb{P}_\varphi$ of the type \eqref{eq:specila_princ}
for a given  map $\varphi\in{\rm Map}(X,\n{U}(2))_{\Z_2}$ which meets the following property: Let $\varphi_i:=\varphi|_{\partial D_i}$ be the restriction of $\varphi$ on the boundary $\partial D_i\simeq\n{S}^2$ of the disk $D_i$ for every $i=1,\ldots,n$. Then, either 
$\varphi_i$ is equivariantly diffeomorphic to the equivariant map $\varphi_\ast:\n{S}^2\to \n{U}(2)$ with antipodal involution defined by
$$
\varphi_\ast(x_1,x_2,x_3)\;:=\;\ii\left(\begin{array}{cc}x_1 & -x_2+\ii x_3 \\x_2+\ii x_3 & x_1\end{array}\right)
$$
or $\varphi_i$ 
 is the constant map at $\n{1}_{\C^2} \in \n{U}(2)$.
\end{lemma}
\proof
Since each connected component $D_i$ of $X_D$ is equivariantly contractible, the principal Q-bundle
 $\bb{P}|_{X_D}$ is trivial. By construction, the involution on $X'$ is free, thus also $\bb{P}|_{X'}$ is trivial as well. This fact follows from 
 \cite[Theorem 4.7 (2)]{denittis-gomi-16} along with the assumption $H^2_{\Z_2}(X,\Z(1))=0$ which implies the triviality of even rank Q-bundles over spaces with free involutions. The  passage
 from vector bundles to principal bundles is then justified by the isomorphism \ref{eq:iso_vect_princ_QR}.
  Let $s_{X_D}$ and $s_{X'}$ be global sections (\ie trivializations) of $\bb{P}|_{X_D}$ and $\bb{P}|_{X'}$, respectively. From these sections one gets the map $\varphi:X'\cap X_D\to\n{U}(2)$ defined by the restriction on $X'\cap X_D$ of the (pointwise) product $s_{X_D}^{-1} s_{X'}$. The map $\varphi$ is equivariant by construction and defines the   principal Q-bundle $\bb{P}_\varphi$ as given in equation \eqref{eq:specila_princ}. The isomorphism $\bb{P}\simeq \bb{P}_\varphi$ is a manifestation of the fact that $\bb{P}$ and $\bb{P}_\varphi$ have the same system of transition functions.
 By the homotopy property of Q-bundles, the Q-isomorphism class of $\bb{P}_\varphi$ only depends on the $\Z_2$-homotopy class of $\varphi$. 
 By \cite[Corollary 4.1]{denittis-gomi-14-gen} 
one has  $[\n{S}^2,\n{U}(2)]_{\Z_2}\simeq\Z_2$ meaning that every equivariant map from the 
 sphere $\n{S}^2$ with the antipodal involution into the space $\n{U}(2)$ with involution $g\mapsto -Q\overline{g} Q$ is $\Z_2$-homotopy equivalent to the constant map at $\n{1}_{\C^2}$ or to the map $\varphi_\ast$. Since $X'\cap X_D$ is a disjoint union of antipodal spheres the map $\varphi$ restricted to each disconnected component can be equivariantly deformed to the constant map at $\n{1}_{\C^2}$ or to the map $\varphi_\ast$. This completes the proof.
\qed

\begin{remark}
{\upshape
 Lemma \ref{lemma:prepar_CS1} deserves two comments. First of all it is worth noticing that the map  $\varphi$ constructed in the proof of the lemma can be always deformed to a smooth map providing in this a way smooth principal Q-bundle $\bb{P}_\varphi$ which represent $\bb{P}$ in the smooth category. This is a manifestation of the equivalence between continuous and smooth category discussed in 
 \cite[Theorem 2.1]{denittis-gomi-15}.
The second observation refers to the content of
Remark \ref{rk:higer-rk}. In fact in view of the  stable rank condition described in Theorem \ref{theo:stab_ran_Q_even} one has that the representation \eqref{eq:specila_princ} must be valid also for principal $\n{U}(2m)$ Q-bundle. In the higher rank case the isomorphism reads
\begin{equation}\label{eq:specila_princ_2m}
\bb{P}\;\simeq\;\bb{P}_\varphi\;:=\;(X' \times \n{U}(2m))\; \bigsqcup_{\varphi'}\; (X_D \times \n{U}(2m))\;
\end{equation}
where the equivariant map ${\varphi'}:X'\cap X_D\to \n{U}(2m)$
factors as
$$
\varphi'\;\simeq\left(\begin{array}{c|c}\varphi & 0 
\\
\hline0 & \n{1}_{\C^{2(m-1)}}\end{array}\right)
$$
and the map ${\varphi}:X'\cap X_D\to \n{U}(2)$ in the upper-left corner meets the properties of Lemma \ref{lemma:prepar_CS1}.
 }\hfill $\blacktriangleleft$
\end{remark}

In view of the Lemma \ref{lemma:prepar_CS1} 
one can assume that $\bb{P}$ is of the form \eqref{eq:specila_princ} since from the beginning.
 With this presentation in hand, the next task is to compute the FKMM-invariant of $\bb{P}$. As a preliminary fact, let us recall that the FKMM-invariant of a principal Q-bundle $(\bb{P},\hat{\Theta})$ is defined as the FKMM-invariant of the associated Q-bundle $(\bb{E},{\Theta})$ (\cf Definition\ref{def:gen_FKMM_inv_princ_bund}).
The FKMM-invariant mesures the 
difference of two trivializations of
the sphere bundle of $\det (\bb{E})|_{X^\tau}$. This is the same as measuring the 
difference of two trivializations of
$\det  (\bb{E})|_{X^\tau}$.

% 
% 
% Let $\phi : X^\tau \to U(1)$ be the equivariant map such that $\phi(x_i) = -1$ if $[\varphi_i] = -1$ in $\Z_2$ and $\phi(x_i) = 1$ if $[\varphi_i] = 1$ in $\Z_2$. (In other words, $\phi(x_i) = \det \varphi_i = \pm 1$.)
%
%

\begin{lemma} \label{lem:FKMM_inv_3mfd}
Assume that the hypotheses of Theorem \ref{prop:Fu-Kane-Mele=chern_d=3} are met. Let $(\bb{P},\hat{\Theta})$ be a
principal $\n{U}(2)$ Q-bundle and $\varphi\in{\rm Map}(X,\n{U}(2))_{\Z_2}$ the equivariant map which represents the principal Q-bundle according to Lemma \ref{lemma:prepar_CS1}. 
Then, the FKMM-invariant of $(\bb{P},\hat{\Theta})$ is represented by the function $\phi:={\rm det} (\varphi)|_{X^\tau}$. More precisely one has that
$$
\kappa(\bb{P},\hat{\Theta})\;=\;[\phi]\;\in\;{\rm Map}(X^\tau,\{\pm1\})\;/\;[X,\n{U}(1)]_{\Z_2}\;.
$$
\end{lemma}
\begin{proof}
 Starting from the representation \ref{eq:specila_princ} one has that
$$
\det(\bb{P})\; =\; (X' \times \n{U}(1))\; \bigsqcup_{\det(\varphi)}\; (X_D \times \n{U}(1))\;.
$$
From this expression one infers that the canonical invariant section $s_{(\bb{P},\hat{\Theta})}$ of $\det(\bb{P})|_{X^\tau}$ is given by
$$
s_{(\bb{P},\hat{\Theta})}\; =\; (x, 1)\; \in\; X^\tau \times \n{U}(1)\; \subset\; \det(\bb{P})\;,
$$
while a global invariant section $s$ of $\det(\bb{P})$ is given by
$$
s(x)\; =\; 
\left\{
\begin{aligned}
&(x, u_{X'}(x))&\qquad &\text{if}\quad x \in X' \\
&(x, u_D(x))&\qquad &\text{if}\quad x \in X_D
\end{aligned}
\right.
$$
where $u_{X'} : X' \to \n{U}(1)$ and $u_D : X_D \to \n{U}(1)$ are two equivariant maps satisfying $u_{X'} = u_D \cdot \det(\varphi)$ on $X' \cap X_D$. According to the presentation of $\bb{P}$ given by \ref{eq:specila_princ} it follows that  $\det (\varphi_i) : X' \cap D_i \to \n{U}(1)$ is a constant map at $1$ or $-1$. Therefore, one can choose $u_{X'}$ to be the constant map at $1$ and $u_D$ to be the locally constant map such that $u_D|_{D_i} = \pm 1$ if $\det (\varphi_i) = \pm 1$.  Then, it follows that
 the FKMM-invariant is represented by $u_D|_{X^\tau} =  \det(\varphi)|_{X^\tau}=:\phi$.
\end{proof}

\medskip

The next goal is to compute the Chern-Simons invariant of $(\bb{P},\hat{\Theta})$. Let $s_{X'}$ and $s_{X_D}$ be the invariant sections of $\bb{P}|_{X'} = X' \times \n{U}(2)$ and $\bb{P}|_{X_D} = X_D \times \n{U}(2)$ defined by
\begin{equation}
\begin{aligned}
s_{X'}(x)\;&=\; (x, \n{1}_{\C^2}) &\qquad&\text{if}\quad x\in X'\\
s_{X_D}(x)\;&=\; (x,  \n{1}_{\C^2}) &\qquad&\text{if}
\quad x\in X_D
\end{aligned}
\end{equation}
respectively.
Then, any section $s$ of $\bb{P}$ is described as
\begin{equation}\label{eq:construct_sect_fromloc}
s(x) =
\left\{
\begin{aligned}
s_{X'}(x) \psi_{X'}(x)^{-1}& \;=\; (x, \psi_{X'}(x)^{-1})&\qquad&\text{if}\quad x \in X' \\
s_{X_D}(x) \psi_{D}(x)^{-1}& \;=\; (x, \psi_{D}(x)^{-1})&\qquad&\text{if}\quad x \in D
\end{aligned}
\right.
\end{equation}
by a pair of maps $\psi_{X'} : X' \to \n{U}(2)$ and $\psi_{D} : X_D \to \n{U}(2)$ such that $\psi_{X'} = \psi_D \varphi$ on $X'\cap D$. 
The map $\psi_{X'}$ and $\psi_D$ can be chosen
smooth in such a way that the  section $s$ is smooth as well. Moreover, the choice of $\psi_{X'}$ and $\psi_{D}$ can be further specified in view of the following result:
\begin{lemma} \label{lem:choose_section}
The smooth maps $\psi_{X'}$ and $\psi_D$ in \eqref{eq:construct_sect_fromloc}
can be chosen
 so that $\psi_D = \n{1}_{\C^2}$ is the constant map.
\end{lemma}
\proof
By construction $\psi_{X'} = \psi_D\varphi$ on $X'\cap D$.
Thus, the proof of the claim reduces to the problem of extending $\varphi : \partial X' \to \n{U}(2)$ to a smooth map $\widetilde{\varphi} : X' \to \n{U}(2)$ so that $\widetilde{\varphi}|_{\partial X'} = \varphi$. Indeed, given such a  $\widetilde{\varphi}$, the proof can be completed by setting $\psi_D = \n{1}_{\C^2}$ and $\psi_{X'} = \widetilde{\varphi}$. 
To prove the existence of $\widetilde{\varphi}$, notice that the three-manifold $X'$ admits a CW decomposition in which the dimension of each cell is at most $3$. The homotopy groups $\pi_i(\n{U}(2))$ are trivial for $i = 0, 2$. The map $\det : \n{U}(2) \to \n{U}(1)$ induces an isomorphism $\pi_1(\n{U}(2)) \simeq \pi_1(\n{U}(1)) \simeq \Z$. Since $\det(\varphi)$ is null-homotpic by construction, one concludes that $\varphi$ extends to a continuous map $\widetilde{\varphi}' : X' \to \n{U}(2)$.
However,  the isomorphism between continuous category and smooth category ensures  the
existence of a smooth map $\widetilde{\varphi} : X' \to \n{U}(2)$, approximating the continuous map $\widetilde{\varphi}'$, that meets  $\widetilde{\varphi}|_{\partial X'} = \varphi$.
\qed

\medskip

Given an invariant connection $\omega$ on $(\bb{P},\hat{\Theta})$, one sets
$$
\omega_{X'} \;:=\; s_{X'}^*\omega\;, \qquad\quad
\omega_{X_D} \;:=\; s_{X_D}^*\omega\;.
$$
The two local expressions are related by
\begin{equation}\label{eq:inv_connect_local}
\omega_{X'}\; =\; \varphi^{-1}\omega_{X_D} \varphi\; +\; \varphi^{-1}d\varphi\;.
\end{equation}
The following result contains  the key computation for the proof of Theorem \ref{prop:Fu-Kane-Mele=chern_d=3}.
\begin{lemma}
\label{lemma:formula_factors}
Assume that the hypotheses of Theorem \ref{prop:Fu-Kane-Mele=chern_d=3} are met. Let $(\bb{P},\hat{\Theta})$ be a
principal $\n{U}(2)$ Q-bundle and $\varphi\in{\rm Map}(X,\n{U}(2))_{\Z_2}$ the equivariant map which represents the principal Q-bundle according to Lemma \ref{lemma:prepar_CS1}. 
Then, the Chern-Simons invariant of $(\bb{P},\hat{\Theta})$ is given by
$$
\rr{cs}(\bb{P},\hat{\Theta})\; =\; \f{WZ}_{\partial X_D}(\varphi) 
\;+\; \frac{1}{8\pi^2} \int_{\partial X_D}{\rm Tr}(\omega_{X_D} \wedge d\varphi\varphi^{-1})\qquad 
{\rm mod.}\;\; \Z
$$
where $\omega_{X_D}$ is defined by \eqref{eq:inv_connect_local} from any invariant section $\omega$.
\end{lemma}
\proof
Let us start with an observation. By construction $\varphi = \bigsqcup_{i=1,\ldots,n} \varphi_i$ and each $\det (\varphi_i) : \partial D_i \to \n{U}(1)$ is constant at $\pm 1$. Hence, $\det(\varphi)$ is null-homotopic and $\f{WZ}_{\partial D}(\varphi)$ makes sense. Now, the computation.
Given the section \eqref{eq:construct_sect_fromloc} one has that
$$
\begin{aligned}
\int_X s^*\f{CS}(\omega)\;&=\;
\int_{X'} s^*\f{CS}(\omega) \;+\; \int_D s^*\f{CS}(\omega) \\
&=\;
\int_{X'} (s_{X'}\psi_{X'}^{-1})^*\f{CS}(\omega) \;+\; \int_D (s_{X_D}\psi_{D}^{-1})^*\f{CS}(\omega)\;.
\end{aligned}
$$
With the help of formula \eqref{eq:formul_CS_W} one has that
$$
 (s_{X'}\psi_{X'}^{-1})^*\f{CS}(\omega)\;=\;s_{X'}^*\f{CS}(\omega)\;+\;\frac{1}{8\pi^2}\; \dd{\rm Tr}\big(s^*_{X'}\omega\wedge \psi_{X'} \dd \psi_{X'}^{-1} \big)\;-\; \frac{1}{24\pi^2}{\rm Tr}\left((\psi_{X'}\dd \psi_{X'}^{-1})^{\wedge 3} \right)\;.
$$
Since
$$
\int_{X'}s_{X'}^*\f{CS}(\omega)\;=\;\int_{X'}\f{CS}(\omega_{X'})\;=\;0
$$
in view of Proposition \ref{propos:Z2_valuesCS_inv} (iii)
one gets
$$
\int_{X'}(s_{X'}\psi_{X'}^{-1})^*\f{CS}(\omega)\;=\;\frac{1}{8\pi^2}\int_{X'} \dd{\rm Tr}\big(\omega_{X'}\wedge \dd\psi_{X'}\psi_{X'}^{-1} \big)\;+\; \f{WZ}_{\partial X'}(\psi_{X'}^{-1}|_{\partial X'})\qquad {\rm mod.}\;\;\Z
$$
where Definition \ref{dfn:WZ_without_boundary} has been used. With a similar computation one gets also
$$
\int_{D}(s_{X_D}\psi_{D}^{-1})^*\f{CS}(\omega)\;=\;\frac{1}{8\pi^2}\int_{D} \dd{\rm Tr}\big(\omega_{X_D}\wedge \dd\psi_{D}\psi_{D}^{-1} \big)\;+\; \f{WZ}_{\partial X_D}(\psi_{D}^{-1}|_{\partial X_D})\qquad {\rm mod.}\;\;\Z
$$
and, after putting all the pieces together, one obtains
$$
\begin{aligned}
\int_X s^*\f{CS}(\omega)\;=\;&\frac{1}{8\pi^2}\int_{X'} \dd{\rm Tr}\big(\omega_{X'}\wedge \dd\psi_{X'}  \psi_{X'}^{-1} \big)\;+\; \frac{1}{8\pi^2}\int_{D} \dd{\rm Tr}\big(\omega_{X_D}\wedge \dd\psi_{D}  \psi_{D}^{-1} \big)\\
&+\;\f{WZ}_{\partial X'}(\psi_{X'}^{-1}|_{\partial X'})\;+\;\f{WZ}_{\partial X_D}(\psi_{D}^{-1}|_{\partial X_D})\qquad {\rm mod.}\;\;\Z.
\end{aligned}
$$
Notice that the orientation on $\partial X'$ induced from $X$ is opposite to that on $\partial D$. Therefore, modulo $\Z$, one gets the following equality
$$
\begin{aligned}
\f{WZ}_{\partial X'}(\psi_{X'}^{-1}|_{\partial X'})\;&=\;-\f{WZ}_{\partial X_D}((\psi_{D}|_{\partial X_D}\varphi)^{-1})\\
&=\;-\f{WZ}_{\partial X_D}(\varphi^{-1})\;-\;
\f{WZ}_{\partial X_D}(\psi_{D}|_{\partial X_D}^{-1})\;-\; \frac{1}{8\pi^2} \int_{\partial D}
{\rm Tr}(\varphi \dd\varphi^{-1} \wedge \dd\psi_D^{-1}\psi_D),
\end{aligned}
$$
which is justified by the relation $\psi_{X'} = \psi_D  \varphi$ on $\partial X' = \partial X_D$ and by the use of the Polyakov-Wiegmann formula proved in Lemma \ref{lem:Polyakov_Wiegmann}.
The local relation between $\psi_{X'}$ and $\psi_D$ also implies
$$
{\rm Tr}\big(\omega_{X'}\wedge \dd\psi_{X'}  \psi_{X'}^{-1} \big)\;=\;
\Tr(\omega_{X_D} \wedge \psi_D^{-1} \dd\psi_D + \omega_{X_D} \wedge \dd\varphi \varphi^{-1}
+ \dd\varphi \varphi^{-1} \wedge \psi_D^{-1}\dd\psi_D).
$$
Summarizing, one finally gets 
$$
\int_X s^*\f{CS}(\omega)\;=\;-\f{WZ}_{\partial X_D}(\varphi^{-1})\;+\; \frac{1}{8\pi^2}\int_X \dd{\rm Tr}(\omega_{X_D} \wedge \dd\varphi \varphi^{-1}) \qquad {\rm mod.}\;\;\Z.
$$
The proof is completed by the general equality
$\f{WZ}_{\partial X_D}(\varphi^{-1}) = - \f{WZ}_{\partial X_D}(\varphi)$ and the use of Definition \ref{def:intrinsC-S_inv}.
\qed 
 
\medskip

We are now in position to provide the proof of the second main result of this work.

\proof[{Proof of Theorem \ref{prop:Fu-Kane-Mele=chern_d=3}}]
Let us choose the maps $\psi_{X'}$ and $\psi_D$ as in Lemma \ref{lem:choose_section}. Then $\omega_{X_D} := s_{X_D}^*\omega=(s\psi_{D})^*\omega = 0$, 
since $\psi_{D}$ is constant. Thus, from the formula in Lemma \ref{lemma:formula_factors} and the definition of the map $\varphi$ one gets
$$
\rr{cs}(\bb{P},\hat{\Theta})\; =\; \f{WZ}_{\partial X_D}(\varphi) 
\;=\;  \sum_{i = 1}^n \f{WZ}_{\partial X_i}(\varphi_i)\qquad {\rm mod.}\;\;\Z.
$$
It holds that $\f{WZ}_{\partial X_i}(\varphi_i)=1$ when $\varphi_i = \n{1}_{\C^2}$ (obvious!) and $\f{WZ}_{\partial X_i}(\varphi_i)=\frac{1}{2}$ when $\varphi_i$ is diffeomorphic to the map $\varphi_\ast$ in Lemma \ref{lemma:prepar_CS1}.
The proof of the latter equality is contained in the proof of
Lemma \ref{lem:surjection_with_WZ}. In fact the map $\varphi_\ast$ coincides with the map \eqref{eq:prefer_chi} and a possible extension $\widetilde{\varphi}_\ast$ on the upper hemisphere  of $\n{S}^3$ can be realized by the prescription \eqref{eq:prefer_chi_ext}. In conclusion one obtains that 
$$
\expo{\ii\; 2\pi\rr{cs}(\bb{P},\hat{\Theta})}\;=\;\Pi({\rm det}(\varphi)|_{X^\tau})\;.
$$
The proof is finally completed by the result in Lemma \ref{lem:FKMM_inv_3mfd}.
\qed

\medskip

Theorem \ref{prop:Fu-Kane-Mele=chern_d=3} has a surprising consequence.

\begin{corollary}
\label{cor:fin_paper}
Under the assumptions in  Theorem \ref{prop:Fu-Kane-Mele=chern_d=3} the homomorphism
$$
\Pi\;:\;{\rm Map}\big(X^\tau,\{\pm 1\}\big)/[X,\n{U}(1)]_{\Z_2}\;\longrightarrow\;\Z_2
$$
induced by the product sign map \eqref{eq:Pi_map} is well-defined.
\end{corollary}
\proof
One needs to shows that the homomorphism $\Pi:{\rm Map}\big(X^\tau,\{\pm 1\}\big)\to\Z_2$
given by the product sign map satisfies ${\Pi}(\phi  \psi|_{X^\tau}) = {\Pi}(\phi)$ for any map $\phi : X^\tau \to \Z_2$ and any equivariant map $\psi : X \to \n{U}(1)$.
Consider the principal $\n{U}(2)$ Q-bundle
$\bb{P}_\varphi$ generated according to \eqref{eq:specila_princ} where the map $\varphi$ is related to $\phi$ as follows: 
$\varphi$ is the constant map at $\n{1}_{\C^2}$ on the disk $D_i$ if $\phi(x_i)=1$ or $\varphi$
agrees with $\varphi_\ast$ on the boundary of 
$D_i$ if $\phi(x_i)=-1$.
By construction the map $\phi$ provides a representative of the FKMM-invariant of 
$\bb{P}_\varphi$ (\cf Lemma \ref{lemma:prepar_CS1}). In a similar way the map $\phi':=\phi  \psi$ represents the FKMM-invariant of an associated principal $\n{U}(2)$ Q-bundles
$\bb{P}_{\varphi'}$. Since $\varphi$ and $\varphi'$
belong to the same class in ${\rm Map}\big(X^\tau,\{\pm 1\}\big)/[X,\n{U}(1)]_{\Z_2}$ it follows that $\bb{P}_\varphi$  and $\bb{P}_{\varphi'}$ have the same FKMM-invariant.  
However, under the hypotheses of  Theorem \ref{prop:Fu-Kane-Mele=chern_d=3} the FKMM-invariant is an isomorphism (\cf Proposition \ref{prop:fkmm-inv_fkmm-space}), hence $\bb{P}_\varphi$  and $\bb{P}_{\varphi'}$ are isomorphic. By using the naturality of the 
 Chern-Simons invariant one gets 
that $\rr{cs}(\bb{P}_{\varphi},\hat{\Theta}_{\varphi})=\rr{cs}(\bb{P}_{\varphi'},\hat{\Theta}_{\varphi'})$. The proof of the claim  then follows in view of formula \eqref{eq:topo_CS_d=3_equality}. 
\qed

%------------------------------------------------------------------------------------------------------------------------%
%                                                                        bibliography
%------------------------------------------------------------------------------------------------------------------------%

\medskip
\medskip

%------------------------------------------------------------------------------------------------------------------------%
\end{document}